\documentclass[]{article}

\usepackage{hyperref}
\usepackage{graphicx}
\usepackage{amsmath}
\usepackage{amssymb}
\usepackage{amsthm}
\usepackage{pxfonts}
\usepackage{enumerate}
\usepackage{color}
\usepackage{mathdots}
\usepackage{mathtools}
\usepackage{sectsty}
\usepackage{amsmath}
\usepackage{nccmath}
\usepackage{tikz}
\usepackage{caption}
\usepackage{adjustbox}
\usepackage{float}
\usepackage{threeparttable}
\usepackage{lipsum}
\usepackage{xcolor}
\usetikzlibrary{matrix,decorations.pathreplacing, calc, positioning,fit}
\usetikzlibrary{arrows.meta}
\usetikzlibrary{shapes.misc}
\usepackage{pgfplots}
\pgfplotsset{compat = newest}
\pgfplotsset{my style/.append style={axis x line=middle, axis y line=
middle, xlabel={$x$}, ylabel={$y$}, axis equal }}
\tikzset{cross/.style={cross out, draw=black, fill=none, minimum size=2*(#1-\pgflinewidth), inner sep=0pt, outer sep=0pt}, cross/.default={2pt}}
\usepackage{subcaption}

\usepackage[title]{appendix}
 \date{}
\title{\large \bf ON THE MOMENTS OF MOMENTS OF RANDOM MATRICES AND EHRHART POLYNOMIALS}

\author{\small THEODOROS ASSIOTIS, EDWARD ERIKSSON AND WENQI NI}

\sectionfont{\scshape\centering\fontsize{11}{14}\selectfont}
\subsectionfont{\scshape\fontsize{11}{14}\selectfont}
\usepackage{fancyhdr}

\newcommand\shorttitle{On the moments of moments of random matrices and Ehrhart polynomials}
\newcommand\authors{Theodoros Assiotis, Edward Eriksson and Wenqi Ni}

\newtheorem{thm}{Theorem}[section]

\newtheorem{lem}[thm]{Lemma}
\newtheorem{defn}[thm]{Definition}
\newtheorem{rmk}[thm]{Remark}
\newtheorem{prop}[thm]{Proposition}

\begin{document}

\maketitle

\begin{abstract}
There has been significant interest in studying the asymptotics of certain generalised moments, called the moments of moments, of characteristic polynomials of random Haar-distributed unitary and symplectic matrices, as the matrix size $N$ goes to infinity. These quantities depend on two parameters $k$ and $q$ and when both of them are positive integers it has been shown that these moments are in fact polynomials in the matrix size $N$. In this paper we classify the integer roots of these polynomials and moreover prove that the polynomials themselves satisfy a certain symmetry property. This confirms some predictions from the thesis of Bailey \cite{EmmaThesis}. The proof uses the Ehrhart-Macdonald reciprocity for rational convex polytopes and certain bijections between lattice points in some polytopes.
\end{abstract}

\tableofcontents

\section{Introduction}

\subsection{History and main results}
We let
\
$$P_{\mathbb{G}(N)}(\theta;\mathbf g) = \det\left(\mathbf I-\mathbf ge^{-i\theta}\right), \ \ \theta \in [0,2\pi],$$ denote the characteristic polynomial, along the unit circle, of a matrix $\mathbf g \in \mathbb{G}\left(N\right)$ where $\mathbb{G}(N)$ is either the group of $N \times N$ unitary matrices $\mathbb{U}(N)$ or the group of $2N\times 2N$ unitary symplectic matrices $\mathbb{SP}(2N)$. We equip these groups with the normalized Haar probability measure and denote by $\mathbb{E}_{\mathbb{G}(N)}$ the expectation with respect to the relevant measure. We also write $\mathsf{\Theta}$ for a uniform random variable on $[0,2\pi]$ and denote by $\mathbf{E}_{\mathsf{\Theta}}$ the corresponding expectation. We then define the following:

\begin{align*}
 \textnormal{MoM}_{\mathbb{G}(N)}(k;q) = \mathbb{E}_{\mathbf{g}\in \mathbb{G}(N)}\left[\left(\mathbf{E}_{\mathsf{\Theta}}\left[|P_{\mathbb{G}(N)}(\mathsf{\Theta};\mathbf{g})|^{2q}\right]\right)^k\right]= \mathbb{E}_{\mathbf{g}\in \mathbb{G}(N)}\left[\left(\frac{1}{2\pi}\int_0^{2\pi}e^{2q \log |P_{\mathbb{G}(N)}(\theta;\mathbf{g})|}d\theta\right)^k\right].   
\end{align*}

These quantities are known as the moments of the moments of the characteristic polynomial $P_{\mathbb{G}(N)}$, or also as the moments of the partition function of the random field $\log |P_{\mathbb{G}(N)}(\cdot;\mathbf{g})|$,
see \cite{FyodorovKeating, FyodorovHiaryKeating,BaileyKeating,ABK,EmmaThesis,CbetaE}.

There has been a lot of interest in the asymptotics of these quantities (and generalisations) as $N$ goes to infinity and by now a great deal is known in the unitary case, see for example \cite{KeatingSnaith, KRRR,  ClaeysKrasovsky, Fahs, BaileyKeating, AssiotisKeating, BGR, CbetaE, KeatingWong} . This study was initiated in papers by Fyodorov-Keating \cite{FyodorovKeating} and Fyodorov-Hiary-Keating \cite{FyodorovHiaryKeating} who made precise conjectures for their asymptotic behaviour and noted that there should be a ``freezing transition" in the asymptotics depending on the value of $kq^2$, see \cite{FyodorovKeating, FyodorovHiaryKeating, KeatingWong}. This problem is also closely related to the theory of certain random measures, called Gaussian Multiplicative Chaos (GMC), see \cite{Berestycki, Webb, NikulaSaksmanWebb, Remy,KeatingWong} and in particular to the moments of the total mass of the GMC on the unit circle, see \cite{FyodorovBouachaud,Remy,KeatingWong}.

The study of the moments of the partition function of $\log|P_{\mathbb{U}(N)}\left(\theta;\mathbf{g}\right)|$ by Fyodorov-Keating \cite{FyodorovKeating} and Fyodorov-Hiary-Keating \cite{FyodorovHiaryKeating} was motivated by considerations of the asymptotic distribution of its maximum on the unit circle, as $N  \to \infty$. These authors made the following remarkable conjecture:
\begin{align*}
 \max_{\theta\in [0,2\pi]}\log|P_{\mathbb{U}(N)}\left(\theta;\mathbf{g}\right)|=\log N-\frac{3}{4}\log\log N+\mathfrak{m}_N(\mathbf{g}),
\end{align*}
where\footnote{Here we slightly abuse notation as $\mathbf{g}$ is actually dependent on $N$; it is Haar-distributed on $\mathbb{U}(N)$.} $\mathfrak{m}_N(\mathbf{g})$ is a sequence of random variables converging in law to a limiting random variable with explicit distribution, see \cite{FyodorovKeating,FyodorovHiaryKeating}. There has been significant progress on this problem, most notably in \cite{ArguinBeliusBourgade,PaquetteZeitouni,ChhaibiMadauleNajnudel}, but the full conjecture is still unresolved.

There are also analogous conjectures of Fyodorov-Keating \cite{FyodorovKeating} and Fyodorov-Hiary-Keating \cite{FyodorovHiaryKeating} for the distribution of the local maxima on the critical line of the Riemann $\zeta$ function , see \cite{NajnudelRiemann,ABBRS,Harper, ABR} for progress on these problems. Moreover, it is possible to define moments of moments of the Riemann $\zeta$ function and these have been studied in \cite{BaileyKeatingRiemann}.

In the symplectic case much less is known currently and as far as we are aware there is no precise conjecture for the asymptotics of the maximum of $\log|P_{\mathbb{SP}(2N)}\left(\theta;\mathbf{g}\right)|$. Nevertheless, the asymptotics of $ \textnormal{MoM}_{\mathbb{SP}(2N)}(k;q)$ for positive integer parameters $k,q \in \mathbb{N}$ were established in \cite{ABK} and connections to GMC in \cite{ForkelKeating}, see also \cite{KeatingWong}.

It is a remarkable fact that, for fixed $k,q\in \mathbb{N}$, the function $N\mapsto \textnormal{MoM}_{\mathbb{G}(N)}(k;q)$ is a polynomial. The following theorem was first proven for the unitary case in \cite{BaileyKeating} and for the symplectic case in \cite{ABK}.

\begin{thm}\label{PolyThm}
Let $k,q \in \mathbb{N}$. Then there exist polynomials $\mathsf{P}_{\left(k;q\right)}^{\mathbb{U}}$ and $\mathsf{P}_{\left(k;q\right)}^{\mathbb{SP}}$, of degrees  $k^2q^2-\left(k-1\right)$ and $kq\left(2kq+1\right)-k$ respectively, such that for all $N\in \mathbb{N}$:
\begin{align*}
  \mathsf{P}_{\left(k;q\right)}^{\mathbb{U}}(N)  &=\textnormal{MoM}_{\mathbb{U}(N)}(k;q),\\
   \mathsf{P}_{\left(k;q\right)}^{\mathbb{SP}}(N) &=\textnormal{MoM}_{\mathbb{SP}(2N)}(k;q).
 \end{align*}
\end{thm}

The proof of polynomiality in the theorem above uses a representation of $\textnormal{MoM}_{\mathbb{G}(N)}(k;q)$ as a multiple integral of a certain symmetric rational function, to which a (large) number of applications of L'H\^{o}pital's rule need to be performed, which after doing the integration gives a polynomial $\mathsf{P}^{\mathbb{G}}_{(k;q)}$ in $N$ (basically one looks at the constant term of this rational function), see \cite{BaileyKeating,ABK}. We note that the argument is rather general and applies to many variants essentially verbatim, but on the other hand is non-explicit. In particular, it is not possible to see what the degree of the polynomial $\mathsf{P}^{\mathbb{G}}_{(k;q)}$ is using this argument and both of the papers above \cite{BaileyKeating,ABK} use different methods (and different between the two) in order to obtain this.

In light of Theorem \ref{PolyThm}, the much-studied problem of asymptotics of $\textnormal{MoM}_{\mathbb{G}(N)}(k;q)$, for $k,q \in \mathbb{N}$, is simply the study of the degree and leading order coefficient of $\mathsf{P}^{\mathbb{G}}_{(k;q)}$. This coefficient is explicit in the unitary case for $k=1$ and $q\in \mathbb{N}$, as proven in \cite{KeatingSnaith}. Moreover, also in the unitary case, for $k=2$ and $q\in \mathbb{N}$, although not explicit, this coefficient is known to have a representation in terms of the Painlev\'{e} V equation, see \cite{ClaeysKrasovsky,BGR}. In the symplectic case much less explicit information is known except for a handful of small values of $(k,q)$ for which the polynomials $\mathsf{P}^{\mathbb{SP}}_{(k;q)}$ have been computed, see \cite{ABK,EmmaThesis}. Thus, to obtain an explicit formula for all the coefficients of these polynomials might seem like a hopeless task. Nevertheless, it is possible to prove non-trivial structural properties of these polynomials and this is the purpose of this paper.

Before stating our main results we give some explicit examples. We note that there is only one infinite family of such polynomials which is explicit. Namely, we have the following formula for $\mathsf{P}^{\mathbb{U}}_{(1;q)}$, first proven by Keating and Snaith in \cite{KeatingSnaith}, see also \cite{BumpGamburd}:
\begin{align}\label{k=1formula}
\mathsf{P}^{\mathbb{U}}_{(1;q)}(N)=\prod_{i=1}^q \prod_{j=1}^q\left(1+\frac{N}{i+j-1}\right).
\end{align}
The next simplest explicit cases are:
\begin{align*}
\mathsf{P}^{\mathbb{U}}_{(2;1)}(N)&=\frac{(N+1)(N+2)(N+3)}{6}, \\ \mathsf{P}^{\mathbb{U}}_{(3;1)}(N)&=\frac{(N+1)(N+2)(N+3)(N+4)(N+5)(N^2+6N+21)}{2520}.
\end{align*}
These have been computed by Bailey \cite{BaileyKeating,EmmaThesis} using a connection to Toeplitz determinants\footnote{This method was suggested by C. Hughes. There are other methods for computing these polynomials but none of them appear to be faster.}. A few more examples can be found in \cite{BaileyKeating,EmmaThesis}. The complexity of the computations rises very rapidly as $k,q$ increase and the most complicated case that has been computed explictly is $(k,q)=(2,3)$, which occupies several lines, see \cite{EmmaThesis}. In the symplectic case things are even more complicated. The following examples have been computed by Bailey \cite{EmmaThesis,ABK} using a connection to Toeplitz+Hankel determinants:
\begin{align*}
\mathsf{P}^{\mathbb{SP}}_{(1;1)}(N)&=\frac{1}{2}(N+1)(N+2),\\ 
\mathsf{P}^{\mathbb{SP}}_{(1;2)}(N)&=\frac{(N+1)(N+2)(N+3)(N+4)(2N+5)(23N^4+230N^3+905N^2+1650N+1512)}{181440},\\
\mathsf{P}^{\mathbb{SP}}_{(2;1)}(N)&=\frac{(N+1)(N+2)(N+3)(N+4)(3N^4+30N^3+127N^2+260N+420)}{10080}.
\end{align*}
The most complicated case that has been computed explicitly is $(k,q)=(2,2)$, see \cite{EmmaThesis}. 

There is a striking feature of these polynomials which is immediately evident: the appearance of the factors $(N+m)$, for $m\in \mathbb{N}$. Based on a few more examples from \cite{EmmaThesis}, Bailey made a precise prediction for the complete list of integer roots of these polynomials. Moreover, after plotting all the complex roots of these polynomials it is remarked in \cite{EmmaThesis} that they should also satisfy some vertical symmetry (in addition to the trivial symmetry about the real axis $\Im(s)=0$). Our two main results confirm these predictions.

\begin{thm}\label{MainThmU} Let $k,q \in \mathbb{N}$. Then $\mathsf{P}_{(k;q)}^{\mathbb{U}}$ has roots at $-1,-2,\dots,-\left(2kq-1\right)$ and it has no more integer roots. Moreover, it satisfies the following symmetry:
\begin{align}\label{unitary_symmetry}\mathsf{P}_{\left(k;q\right)}^{\mathbb{U}}\left(-kq+s\right)=\left(-1\right)^{k^2q^2-\left(k-1\right)}\mathsf{P}_{\left(k;q\right)}^{\mathbb{U}}\left(-kq-s\right), \ \ s \in \mathbb{C}.\end{align}

\end{thm}

\begin{thm}\label{MainThmSp}
Let $k,q \in \mathbb{N}$. Then  $\mathsf{P}_{\left(k;q\right)}^{\mathbb{SP}}$ has roots at $-1,-2,\dots,-2kq$ and it has no more integer roots. Moreover, it satisfies the following symmetry:
\begin{align}\label{symplectic_symmetry}\mathsf{P}_{(k;q)}^{\mathbb{SP}}\left(-kq-\frac{1}{2}+s\right)=\left(-1\right)^{kq\left(2kq+1\right)-k}\mathsf{P}_{\left(k;q\right)}^{\mathbb{SP}}\left(-kq-\frac{1}{2}-s\right), \ \ s \in \mathbb{C}.\end{align}
\end{thm}

\begin{rmk}
In the case $k=1$ and $q\in \mathbb{N}$, Theorem \ref{MainThmU} can be be readily checked using the explicit formula (\ref{k=1formula}) for $\mathsf{P}_{(1;q)}^{\mathbb{U}}$.
\end{rmk}

\begin{rmk}
By looking at the formula (\ref{k=1formula}) for $\mathsf{P}_{(1;q)}^{\mathbb{U}}$ we can observe that the integer roots of these polynomials can come with multiplicities. It would be interesting to determine these and explore whether there is a combinatorial interpretation for them. Moreover, by looking at $\mathsf{P}_{(1;2)}^{\mathbb{SP}}$ it can be seen that these polynomials can have rational (non-integer) roots. It would be interesting if these can be determined as well.
\end{rmk}

\subsection{Strategy of proof, connections and extensions}

The starting point in our analysis is a representation of $\textnormal{MoM}_{\mathbb{G}(N)}(k;q)$ in terms of the number of lattice points in some kind of Gelfand-Tsetlin type \cite{AssiotisKeating,ABK} polytope $\mathcal{V}^{\mathbb{G}}_{(k;q)}$ dilated by a factor of $N$. We then use the Ehrhart-Macdonald reciprocity \cite{EhrhartPoly1,EhrhartPoly2,Macdonald,StanleyBook,CombinatorialReciprocity} for rational convex polytopes to access the values of the polynomials $\mathsf{P}_{(k;q)}^{\mathbb{G}}$ on the negative integers. In order to prove the symmetry we find some natural bijections between the lattice points in a dilation of the interior of $\mathcal{V}^{\mathbb{G}}_{(k;q)}$ and the lattice points in a less dilated version of the polytope $\mathcal{V}^{\mathbb{G}}_{(k;q)}$ itself. This is also closely related to the notion of reflexive, and more generally Gorenstein, polytopes, see Remark \ref{RemarkReflexive} for more details.

We also note that, as proven in Section \ref{SectionPrelim}, the polytopes $\mathcal{V}_{(k;q)}^{\mathbb{G}}$ are in general, except for a few cases, not integral polytopes; they are simply rational. In particular, the polynomiality statement in Theorem \ref{PolyThm} is not an immediate consequence of Ehrhart theory \cite{StanleyBook,CombinatorialReciprocity}. It is well-known that the lattice point counting functions for rational convex polytopes are in general quasi-polynomials and not necessarily polynomials, see \cite{StanleyBook,CombinatorialReciprocity}. Nevertheless, there are many examples (and we consider some infinite families in this paper) of non-integral rational polytopes with polynomial lattice point counting functions. This phenomenon is known as period collapse, see for example \cite{McAllisterWoods,Woods}, and it is far from being completely understood. Despite it not being a direct consequence of Ehrhart theory, it is still possible to give a purely combinatorial proof of the polynomiality result in the unitary case, and we do this in Section \ref{SectionPrelim}.

In some sense, the prototypical result of the kind we prove here is a well-known theorem of Stanley and Ehrhart on the counting function of magic squares, see \cite{StanleyMagic,EhrhartMagic,StanleyBook}. Magic squares are simply the lattice points in dilations of the famous Birkhoff polytopes\footnote{It is well-known that these are integral polytopes as their vertices are the permutation matrices, see \cite{StanleyBook}. } of doubly-stochastic matrices, see \cite{StanleyBook, BeckPixton}.  Magic squares, and generalisations, also have applications in random matrix theory, in particular to the moments of secular coefficients of random matrix characteristic polynomials, see \cite{DiaconisGamburd,ConreyGamburd,NPS}. They also come up in the study of particular matrix integrals which arise from considerations of number theoretic problems over function fields, see \cite{KRRR}.

We close this introduction with some remarks.

\begin{rmk}[On the connection to reflexive \& dual-integral polytopes]\label{RemarkReflexive}
The notion of a reflexive polytope was introduced by Batyrev in the context of mirror symmetry in algebraic geometry, see \cite{Batyrev}. This has a few equivalent definitions, the original being in terms of the dual polytope, see \cite{Batyrev}. There is also a more general notion of a Gorenstein polytope of some index $r$ ($r=1$ corresponds to being reflexive), see \cite{CombinatorialReciprocity}. For a lattice polytope which contains the origin in its interior, the kind of symmetry property (\ref{unitary_symmetry}), (\ref{symplectic_symmetry}), we prove in this paper is a necessary and sufficient condition for the polytope to be Gorenstein, see \cite{CombinatorialReciprocity}. However, our polytopes are in general not integral, see Propositions \ref{NonIntegral} and \ref{NonIntegralSymplectic} and thus are not Gorenstein polytopes. On the other hand, from looking at the simplest possible non-integral case (which is $\mathcal{V}_{(2;1)}^{\mathbb{U}}$, see Figure \ref{fig:pyramid}), one might guess that $2kq\mathcal{V}_{(k;q)}^{\mathbb{U}}$ and $(2kq+1)\mathcal{V}_{(k;q)}^{\mathbb{SP}}$ are integral polytopes, in which case\footnote{After an appropriate translation to include the origin in their interior.} our symmetry property would imply that they are reflexive, see \cite{HibiDuality}. However, this is also not true in general and this follows from Proposition \ref{NonIntegralU+} and the form of the non-integral vertex in Proposition \ref{NonIntegralSymplectic}. The correct notion describing these polytopes is that of a dual-integral polytope as given in \cite{Reinert}.
\end{rmk}

\begin{rmk} [The orthogonal case] \label{RemarkOrthogonal}
It is possible to consider moments of moments over the special orthogonal group $\mathbb{SO}(2N)$ of $2N\times2N$ orthogonal matrices with determinant $+1$. There is an analogue of Theorem \ref{PolyThm}, see \cite{ABK}, as well. The corresponding polynomial $\mathsf{P}_{(k;q)}^{\mathbb{SO}}$ is of degree $kq(2kq-1)-k$, except in the exceptional case $(k,q)=(1,1)$ when it is of degree $1$. Again, this has a representation in terms of counting lattice points, not in a single rational convex polytope, but in a rational polyhedral complex. Now, it is known that Ehrhart-Macdonald reciprocity extends to some general polyhedral complexes under certain conditions, see for example \cite{HibiPolyhedral,HibiBook}. We believe that a generalisation\footnote{These results are stated for lattice polyhedral complexes there.} of these results should be applicable to the orthogonal case, but proving this carefully would involve some subtle technical issues and we do not pursue it in this paper. Assuming this claim, an adaptation of the arguments presented in the symplectic case should give the following analogue of our main results, for $(k,q)\neq (1,1)$: the complete list of integer roots of $\mathsf{P}_{(k;q)}^{\mathbb{SO}}$ is given by $\{-1,-2,\dots, -(2kq-2) \}$ and moreover $\mathsf{P}_{(k;q)}^{\mathbb{SO}}$ satisfies the following symmetry:
\begin{align*}
\mathsf{P}_{(k;q)}^{\mathbb{SO}}\left(-kq+\frac{1}{2}+s\right)=\left(-1\right)^{kq\left(2kq-1\right)-k}\mathsf{P}_{\left(k;q\right)}^{\mathbb{SO}}\left(-kq+\frac{1}{2}-s\right), \ \ s \in \mathbb{C}.
\end{align*}
\end{rmk}

\begin{rmk}
It would be interesting to understand the structure of the general roots of these polynomials. Some general properties and conjectures for roots of Ehrhart polynomials of lattice (in particular they do not directly apply to our setting) polytopes have appeared in \cite{RootsEhrhart1,RootsEhrhart2}.
\end{rmk}

\paragraph{Organisation of the paper}
In Section \ref{SectionPrelim} we collect some previous results which connect the problem at hand to the study of some rational convex polytopes. Additionally, we recall some classical results from Ehrhart theory. We moreover establish some basic properties of these polytopes and give an alternative combinatorial proof of the polynomiality theorem in the unitary case.
In Section 3 we prove Theorem \ref{MainThmU} . We follow an analogous strategy in Section 4 and prove Theorem \ref{MainThmSp}.

\paragraph{Acknowledgements}
We are grateful to the School of Mathematics of the University of Edinburgh for financial support through vacation scholarships.
We thank P. Zjaikin for taking part in some initial meetings for this project and are grateful to EPSRC for financial support through a vacation internship. We are grateful to E. Bailey for discussions and useful comments on a first draft of this paper. Finally, we are grateful to the referee for a careful reading of the paper and useful comments and suggestions which have improved the presentation. The tikz code for Figure \ref{fig:pyramid} was produced using $\mathsf{SageMath}$.

\section{Preliminaries}\label{SectionPrelim}

\subsection{General preliminaries}
For parameters $k,q\in \mathbb{N}$, it transpires that the moments of moments can be studied for $N \in \mathbb{N}$ by counting certain combinatorial objects. In the following subsections we will define those objects (and also prove some results of independent interest about them) and lay the foundation for extending to $N \in \mathbb{Z}$. We begin with some notation and definitions.

\begin{defn}
Let $\mathcal{T} \subseteq \mathbb{R}$ and $b \in \mathbb{R}$. Then, we define $\mathcal{T}_{\geq b} = \{t \in \mathcal{T}: t \geq b\}$ and define $\mathcal{T}_{>b}$ analogously.
\end{defn}

\begin{defn}
A signature $\boldsymbol\lambda$ of length $M$ is a tuple of $M$ integers where we have $\lambda_1 \geq \cdots \geq \lambda_M$. Let $\mathcal{S}_M^+$ be the set of all signatures of length $M$ with non-negative elements. For $\boldsymbol \lambda \in \mathcal{S}_M^+$ we define:
$$\left|\boldsymbol\lambda\right| = \lambda_1 + \cdots + \lambda_M.$$
We will use $\#$ for the number of elements in a set.
\end{defn}

We will also make frequent use of the following continuous analogue to the above definition.

\begin{defn} Let $M\in \mathbb{N}$.
Define the non-negative Weyl chamber of length $M$ by
$$\mathcal{W}_M^+ = \left\{(w_1,\ldots,w_M) \in \mathbb{R}^{M}: w_1 \geq \cdots \geq w_M \geq 0 \right\}.$$
As for signatures, we use the notation $|\cdot|$ for the sum of coordinates.
\end{defn}

Furthermore we need the classical notion of interlacing:

\begin{defn}
\label{interlacing}
A signature $\boldsymbol\mu \in \mathcal{S}_{M}^+$ is said to interlace with $\boldsymbol\lambda \in \mathcal{S}_{M-1}^+$ if the following inequalities are satisfied: $$\lambda_1 \geq \mu_1 \geq \lambda_2 \ge \cdots \geq \mu_{M-1} \geq \lambda_M.$$
A signature $\boldsymbol\mu \in \mathcal{S}_{M}^+$ is said to interlace with $\boldsymbol\lambda \in \mathcal{S}_{M}^+$ if the following inequalities are satisfied:

$$\lambda_1 \geq \mu_1 \geq \lambda_2 \ge \cdots \geq \lambda_{M} \geq \mu_{M}.$$
In either of these cases we write $\boldsymbol\mu \prec \boldsymbol\lambda$. If strict versions of the above inequalities are satisfied we say that the interlacing is strict. We define interlacing for elements of $\mathcal{W}_M^+$ in an entirely analogous fashion.
\end{defn}

We now recall some standard definitions and facts about rational convex polytopes and the associated Ehrhart theory. For a more detailed treatment see \cite{Ziegler,StanleyBook,CombinatorialReciprocity}.

\begin{defn}
A convex polytope $\mathcal{V} \subset \mathbb{R}^m$ is the set of solutions to a finite linear system of inequalities, provided that the solutions are bounded. Namely, $\mathcal{V} = \left\{\boldsymbol x \in \mathbb{R}^m: \boldsymbol \alpha \boldsymbol x \leq \boldsymbol a \right\},$ where $\boldsymbol \alpha$ is a matrix of real coefficients, $\boldsymbol a$ is a vector with real entries and a vector inequality holds if and only if inequality holds in each element. It is always the case that $\mathcal{V}$ is homeomorphic to a $d$-ball $\mathbb{B}^d$ and so we define the dimension of $\mathcal{V}$ by $\textnormal{dim}\left(\mathcal{V}\right) = d$.
\end{defn}

\begin{defn}
A face $\mathcal{F}$ of a polytope $\mathcal{V}$ is any subset of the form $\mathcal{F} = \{\boldsymbol v \in \mathcal{V}: \boldsymbol{\beta}^t\boldsymbol{v} = b \}$ where $b \in \mathbb{R}$, $\boldsymbol \beta \in \mathbb{R}^m$, and where for all $ \boldsymbol v \in \mathcal{V}$ we have $\boldsymbol{\beta}^t\boldsymbol{v} \leq b $. If a face of $\mathcal{V}$ is $0$-dimensional we call it a vertex of $\mathcal{V}$. If all coordinates of all vertices of $\mathcal{V}$ are rational we say that $\mathcal{V}$ is a rational polytope. We define integral or lattice polytopes analogously.
\end{defn}

\begin{defn}
\label{quasi_polynomial}
A function $g: \mathbb{Z} \rightarrow \mathbb{C}$ or $g: \mathbb{Z}_{\geq 0} \rightarrow \mathbb{C}$ is a quasi-polynomial if there exist polynomials $g_1,\ldots,g_{D\left(g\right)}$, with $D(g) \in \mathbb{N}$ depending only on $g$,  such that: $$g(x) = g_i(x), \textnormal{ whenever } x \equiv i \textnormal{ mod } D(g).$$
\end{defn}

We also need the following little observation.

\begin{lem}
\label{vanishing_quasi}
Let $g: \mathbb{Z} \rightarrow \mathbb{C}$ be a quasi-polynomial such that for all $x \in \mathbb{N}$ we have
$g(x) = 0$. 
Then, $g$ is identically zero on $\mathbb{Z}$.
\end{lem}

\begin{proof}
Since $g$ is a quasi-polynomial then by Definition \ref{quasi_polynomial}, each $g_i$ vanishes at infinitely many points and is thus identically zero.
\end{proof}

Throughout this paper we denote by  $\textnormal{int}\left(\mathcal{A}\right)$ the interior of a set $\mathcal{A}$ in Euclidean space. Moreover, we denote by $N\mathcal{A}$ the dilate of $\mathcal{A}$ by a factor of $N$, namely $N\mathcal{A} = \{Nv: v \in \mathcal{A}\}$. We then have the following classical result of Ehrhart, see \cite{EhrhartPoly1,EhrhartPoly2,StanleyBook}. 

\begin{thm}
\label{Ehrhart_polynomials}
Let $\mathcal{V} \subset \mathbb{R}^m$ be a rational convex polytope. Define the following quantities for any $N\in \mathbb{Z}_{\geq 0}$:
\begin{align*}
\mathcal{L}\left(\mathcal{V},N\right) 
&= \#\left(N\mathcal{V} \cap \mathbb{Z}^m\right),\\
\overline{\mathcal{L}}\left(\mathcal{V},N\right) 
&= \#\left(N \textnormal{int}\left(\mathcal{V}\right) \cap \mathbb{Z}^m\right).
\end{align*}
Then, $\mathcal{L}\left(\mathcal{V},\cdot\right)$ and $\overline{\mathcal{L}}\left(\mathcal{V},\cdot\right)$ are quasi-polynomials. Moreover, if $\mathcal{V}$ is an integral convex polytope they are polynomials.
\end{thm}

\begin{rmk}
As alluded to in the introduction, it is not necessary for $\mathcal{V}$ to be an integral convex polytope for the associated function $\mathcal{L}$ to be a bona fide polynomial. In fact, in Proposition \ref{NonIntegral} and Proposition \ref{NonIntegralSymplectic} below we prove that the infinite families of polytopes we consider here are examples of this phenomenon.
\end{rmk}

$\mathcal{L}\left(\mathcal{V},\cdot\right)$ and $\overline{\mathcal{L}}\left(\mathcal{V},\cdot\right)$ can be uniquely extended as quasi-polynomials for $N \in \mathbb{Z}$ by Lemma \ref{vanishing_quasi}. We abuse notation and use the same symbols for these extensions. The following theorem, see \cite{EhrhartPoly1,EhrhartPoly2,Macdonald,StanleyBook}, the so-called Ehrhart-Macdonald reciprocity (this was conjectured and partially proven by Ehrhart and in full generality proven by Macdonald \cite{Macdonald}), will play a key role in allowing us to extend the moments of moments to $N \in \mathbb{Z}$. 
\begin{thm}
\label{ehrhart_macdonald}
Let $N \in \mathbb{N}$ and let $\mathcal{V}$ be a rational convex polytope with $\textnormal{dim}\left(\mathcal{V}\right) = d$. Then, $$\mathcal{L}\left(\mathcal{V},-N\right) = \left(-1\right)^d \overline{\mathcal{L}}\left(\mathcal{V},N\right).$$ \end{thm}

\subsection{Unitary case preliminaries}
We begin by defining the combinatorial objects that allow us to evaluate $\textnormal{MoM}_{\mathbb{U}(N)}(k;q)$ for $k,q,N \in \mathbb{N}$. 
\begin{defn}\label{unitary_def}
A unitary Gelfand-Tsetlin pattern with $M$ rows is a sequence of $M$ signatures $\left(\boldsymbol\lambda^{\left(j\right)}\right)_{j=1}^{M}$ where\footnote{The standard definition does not involve the non-negativity constraint. In this paper we will only consider non-negative signatures.} $\boldsymbol\lambda^{(j)} \in \mathcal{S}_j^+$  and $$\boldsymbol\lambda^{\left(1\right)} \prec \boldsymbol\lambda^{\left(2\right)} \prec \boldsymbol\lambda^{\left(3\right)} \prec \cdots \prec \boldsymbol\lambda^{\left(M\right)}.$$
If instead $\boldsymbol\lambda^{(j)} \in \mathcal{W}_j^+$ we say that $\left(\boldsymbol\lambda^{\left(j\right)}\right)_{j=1}^{M}$ is a unitary continuous Gelfand-Tsetlin pattern. In either case, if the interlacing is strict we say that it is a strict pattern. 
\end{defn}

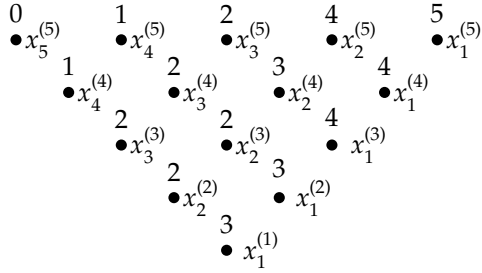
\begin{figure}[H]
\centering
\begin{tikzpicture}[scale = 0.7]
%bullets
\fill (0,0) circle (3pt) node[right=0.1cm] {$x_1^{(1)}$} node[above=0.1cm] {3};
\fill (1,1) circle (3pt) node[right=0.1cm] {$x_1^{(2)}$} node[above=0.15cm] {3};
\fill (1,3) circle (3pt) node[right] {$x_2^{(4)}$} node[above=0.1cm] {3};
\fill (2,2) circle (3pt) node[right=0.13cm] {$x_1^{(3)}$} node[above=0.13cm] {4};
\fill (2,4) circle (3pt) node[right] {$x_2^{(5)}$} node[above=0.1cm] {4};
\fill (3,3) circle (3pt) node[right] {$x_1^{(4)}$} node[above=0.1cm] {4};
\fill (-4,4) circle (3pt) node[right] {$x_5^{(5)}$} node[above=0.1cm] {0};
\fill (-2,4) circle (3pt) node[right] {$x_4^{(5)}$} node[above=0.1cm] {1};
\fill (-3,3) circle (3pt) node[right] {$x_4^{(4)}$} node[above=0.1cm] {1};
\fill (-1,1) circle (3pt) node[right] {$x_2^{(2)}$} node[above=0.1cm] {2};
\fill (-2,2) circle (3pt) node[right] {$x_3^{(3)}$} node[above=0.1cm] {2};
\fill (0,2) circle (3pt) node[right] {$x_2^{(3)}$} node[above=0.1cm] {2};
\fill (0,4) circle (3pt) node[right] {$x_3^{(5)}$} node[above=0.1cm] {2}; 
\fill (-1,3) circle (3pt) node[right] {$x_3^{(4)}$} node[above=0.1cm] {2};
\fill (4,4) circle (3pt) node[right] {$x_1^{(5)}$} node[above=0.1cm] {5};
\end{tikzpicture}
\captionsetup{justification=raggedright,singlelinecheck=off}
\caption{An example of a unitary Gelfand-Tsetlin pattern with $M=5$ rows.}
\label{fig:unitaryGT_pattern}
\end{figure}

We have the following important definition:

\begin{defn}\label{UnitaryPatterns}
Let $k,q,N \in \mathbb{N}$. We define $\mathcal{UP}_{\left(k;q;N\right)} \subset \mathbb{Z}^{k^2q^2}$, consisting of elements $\left(p_i^{(j)}\right)$, as follows:
\begin{enumerate}
\item For all $(i,j) \in \left\{(m,n) \in \mathbb{N}^2: n \leq 2kq-1, m \leq kq-\left|kq-n\right| \right\}$ we have $0 \leq  p_{i}^{\left(j\right)} \leq N.$
\item ${\left(\boldsymbol p^{\left(j\right)}\right)}_{j=1}^{kq}$ and ${\left(\boldsymbol p^{\left(2kq-j\right)}\right)}_{j=1}^{kq}$ form unitary Gelfand-Tsetlin patterns.
\item For $j \in \{1,\ldots,k-1\}$:
$\left|\boldsymbol p^{\left(2jq\right)}\right| = N\frac{kq-\left|kq-2jq\right|}{2}.$
\end{enumerate}
We will refer to the $k-1$ constraints in Condition 3 as the sum constraints. Denote the analogous set of integer arrays with strict interlacing and entries strictly between $0$ and $N$ by $\mathcal{UP}_{\left(k;q;N\right)}^{\neq}$. \end{defn} The following result was proven in Section 2 of  \cite{AssiotisKeating}.

\begin{prop}
\label{unitary_patterns}
Let $k,q \in\mathbb{N}$. Then for all $N \in \mathbb{N}$ we have: $$\textnormal{MoM}_{\mathbb{U}(N)}\left(k;q\right) = \#\mathcal{UP}_{\left(k;q;N\right)}.$$
\end{prop} We will now move from counting patterns to counting lattice points in the dilates of an appropriate polytope. Observe that, since the definition of $\mathcal{UP}_{\left(k;q;N\right)}$ has $k-1$ sum constraints we can determine $k-1$ entries from the others. We call these constraints the sum constraints. We define the following index set to focus on the $k^2q^2-\left(k-1\right)$ free variables.
We let $$\mathcal{S}_{\left(k;q\right)}^{\mathbb{U}} = \left\{\left(i,j\right) \in \mathbb{N}^2:j \leq 2kq-1, i \leq kq-\left|kq-j\right|,\left(i,j\right) \neq\left(1,2ql\right) \textnormal{ for }l =1,\ldots,k-1\right\},$$
and note that $\mathcal{S}_{\left(k;q\right)}^{\mathbb{U}}$ has $k^2q^2-\left(k-1\right)$ elements. We then define the following:

\begin{defn}\label{unitary_polytope}
Let $k,q\in \mathbb{N}$. Define $\mathcal{V}^{\mathbb{U}}_{\left(k;q\right)} \subset \mathbb{R}^{k^2q^2-\left(k-1\right)}$, consisting of elements $\left(v_i^{(j)}\right)$, by: \begin{enumerate}
    \item For all $\left(i,j\right) \in \mathcal{S}_{\left(k;q\right)}^{\mathbb{U}}$ we have
$0 \leq v_i^{\left(j\right)} \leq 1.$
\item For $l = 1,\ldots,k-1$ we define additional elements by: \begin{align}\label{unitary_fixed} v_1^{(2ql)} = \frac{kq-\left|kq-2ql\right|}{2}- v_2^{(2ql)}-\cdots-v_{kq-\left|kq-2ql\right|}^{(2ql)}, \end{align} and require that $0 \leq v_1^{\left(2ql\right)} \leq 1.$
\item $\left(\boldsymbol v^{\left(j\right)}\right)_{j=1}^{kq}$ and $\left(\boldsymbol v^{\left(2kq-j\right)}\right)_{j=1}^{kq}$ form continuous unitary Gelfand-Tsetlin patterns.
\end{enumerate}
\end{defn} 

%pyramid 
\begin{figure}[H]
\centering
\resizebox{0.4\linewidth}{!}{
\begin{tikzpicture}%
[x={(-0.925084cm, -0.203511cm)},
y={(0.379763cm, -0.495689cm)},
z={(-0.000024cm, 0.844320cm)},
scale=5.000000,
back/.style={loosely dotted, thin},
edge/.style={color=blue!95!black, thick},
facet/.style={fill=blue!95!black,fill opacity=0.250000},
vertex/.style={inner sep=1pt,circle,draw=green!25!black,fill=green!75!black,thick}]
%
%
%% This TikZ-picture was produce with Sagemath version 9.4
%% with the command: ._tikz_3d_in_3d and parameters:
%% view = [-0.0946000000000000, -0.479600000000000, -0.872400000000000]
%% angle = 160.470000000000
%% scale = 1
%% edge_color = blue!95!black
%% facet_color = blue!95!black
%% opacity = 0.250000000000000
%% vertex_color = green
%% axis = True

%% Drawing the axes
\draw[color=black,thick,->] (0,0,0) -- (1,0,0) node[anchor=north east]{$x$};
\draw[color=black,thick,->] (0,0,0) -- (0,1,0) node[anchor=north west]{$y$};
\draw[color=black,thick,->] (0,0,0) -- (0,0,0.75) node[anchor=south]{$z$};
%% Coordinate of the vertices:
%%
\coordinate (0.00000, 0.00000, 0.00000) at (0.00000, 0.00000, 0.00000);
\coordinate (0.00000, 1.00000, 0.00000) at (0.00000, 1.00000, 0.00000);
\coordinate (1.00000, 0.00000, 0.00000) at (1.00000, 0.00000, 0.00000);
\coordinate (1.00000, 1.00000, 0.00000) at (1.00000, 1.00000, 0.00000);
\coordinate (0.50000, 0.50000, 0.50000) at (0.50000, 0.50000, 0.50000);
%%
%%
%% Drawing edges in the back
%%
\draw[edge,back] (0.00000, 0.00000, 0.00000) -- (1.00000, 0.00000, 0.00000);
%%
%%
%% Drawing vertices in the back
%%
%%
%%
%% Drawing the facets
%%
\fill[facet] (0.50000, 0.50000, 0.50000) -- (0.00000, 0.00000, 0.00000) -- (0.00000, 1.00000, 0.00000) -- cycle {};
\fill[facet] (0.50000, 0.50000, 0.50000) -- (1.00000, 0.00000, 0.00000) -- (1.00000, 1.00000, 0.00000) -- cycle {};
\fill[facet] (0.50000, 0.50000, 0.50000) -- (0.00000, 1.00000, 0.00000) -- (1.00000, 1.00000, 0.00000) -- cycle {};
%%
%%
%% Drawing edges in the front
%%
\draw[edge] (0.00000, 0.00000, 0.00000) -- (0.00000, 1.00000, 0.00000);
\draw[edge] (0.00000, 0.00000, 0.00000) -- (0.50000, 0.50000, 0.50000);
\draw[edge] (0.00000, 1.00000, 0.00000) -- (1.00000, 1.00000, 0.00000);
\draw[edge] (0.00000, 1.00000, 0.00000) -- (0.50000, 0.50000, 0.50000);
\draw[edge] (1.00000, 0.00000, 0.00000) -- (1.00000, 1.00000, 0.00000);
\draw[edge] (1.00000, 0.00000, 0.00000) -- (0.50000, 0.50000, 0.50000);
\draw[edge] (1.00000, 1.00000, 0.00000) -- (0.50000, 0.50000, 0.50000);
%%
%%
%% Drawing the vertices in the front
%%
\node[vertex,scale=2pt] at (0.00000, 0.00000, 0.00000) {};
\node[vertex,scale=2pt] at (0.00000, 1.00000, 0.00000) {};
\node[vertex,scale=2pt] at (1.00000, 0.00000, 0.00000) {};
\node[vertex,scale=2pt] at (1.00000, 1.00000, 0.00000) {};
\node[vertex,scale=2pt] at (0.50000, 0.50000, 0.50000) {};
%%add notation
\node[above right] (a) at (0.00000, 0.00000, 0.00000) {(0,0,0)};
\node[above right] (b) at (0.00000, 1.00000, 0.00000) {(0,1,0)};
\node[above left] (c) at (1.00000, 0.00000, 0.00000) {(1,0,0)};
\node[below left] (d) at (1.00000, 1.00000, 0.00000) {(1,1,0)};
\node[above=0.3cm] (e) at (0.50000, 0.50000, 0.50000) {($\frac{1}{2}$,$\frac{1}{2}$,$\frac{1}{2})$};
\end{tikzpicture}}
\captionsetup{justification=raggedright,singlelinecheck=off}
\caption{A visualisation of $\mathcal{V}^{\mathbb{U}}_{\left(2;1\right)}$. It is simply a pyramid with vertices at $(0,0,0)$, $(0,1,0)$, $(1,1,0)$, $(1,0,0)$, $\left(\frac{1}{2},\frac{1}{2},\frac{1}{2}\right)$. We note the existence of a non-integral vertex.}
\label{fig:pyramid}
\end{figure}
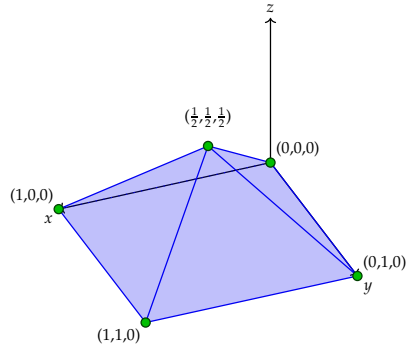

\begin{figure}[H]
\centering
\begin{tikzpicture}[scale = 0.8]
\draw[color=black!80, thick] (-2,0)--(0,-4)--(2,0)--(0,4)--(-2,0);
\draw[color=black!40,ultra thick] (-1,2)--(1,2);
\draw[color=black!40,ultra thick] (-2,0)--(2,0);
\draw[color=black!40,ultra thick] (-1,-2)--(1,-2);
%node
\fill[red] (0,3.8) circle (1.5pt); 
\fill[red] (-0.1,3.6) circle (1.5pt);
\fill[red] (0.1,3.6) circle (1.5pt);
\fill[red] (-0.2,3.4) circle (1.5pt);
\fill[red] (0,3.4) circle (1.5pt);
\fill[red] (.2,3.4) circle (1.5pt);

\fill[red] (0,-3.8) circle (1.5pt); 
\fill[red] (-0.1,-3.6) circle (1.5pt);
\fill[red] (0.1,-3.6) circle (1.5pt);
\fill[red] (-0.2,-3.4) circle (1.5pt);
\fill[red] (0,-3.4) circle (1.5pt);
\fill[red] (.2,-3.4) circle (1.5pt);

\fill[red] (-0.9,2) circle (1.5pt);
\fill[red] (-0.6,2) circle (1.5pt);
\fill[red] (0.6,2) circle (1.5pt);
\draw (0.9,2) node[line width=0.2mm, cross=2.5pt,blue!85] {} node[above right, blue!85] {$v_1^{(6q)}$};
\node[red,thick] (a) at (0,2.1) {\large \ldots};

\fill[red] (-1.9,0) circle (1.5pt);
\fill[red] (-1.7,0) circle (1.5pt);
\draw (1.9,0) node[line width=0.2mm, cross=2.5pt,blue!85] {} node[above right, blue!85] {$v_1^{(4q)}$};
\fill[red] (1.7,0) circle (1.5pt);
\node[red,thick] (b) at (0,0.1) {\large\ldots};
\node[red,thick] (c) at (-1,0.1) {\large\ldots};
\node[red,thick] (d) at (1,0.1) {\large\ldots};

\fill[red] (-0.9,-2) circle (1.5pt);
\fill[red] (-0.6,-2) circle (1.5pt);
\fill[red] (0.6,-2) circle (1.5pt);
\draw (0.9,-2) node[line width=0.2mm, cross=2.5pt,blue!85] {} node[below right, blue!85] {$v_1^{(2q)}$};
\node[red] (c) at (0,-1.9) {\large\ldots};

%notation
\node (a) at (3.3,2) {$6q$-th row};
\node (b) at (3.3,0) {$4q$-th row};
\node (c) at (3.3,-2) {$2q$-th row};

\draw[black] (a)--(1,2);
\draw[black] (b)--(2,0);
\draw[black] (c)--(1,-2);

\end{tikzpicture}
\captionsetup{justification=raggedright,singlelinecheck=off}
\caption{Illustration of the fixed and free variables from Definition \ref{unitary_polytope} with $k=4$ and a general $q$. The free variables are depicted as red circles and the fixed variables as blue crosses. Note that, there are exactly $k-1=3$ fixed variables which is simply the number of sum constraints.}
    \label{fig:unitary_k4}
\end{figure}
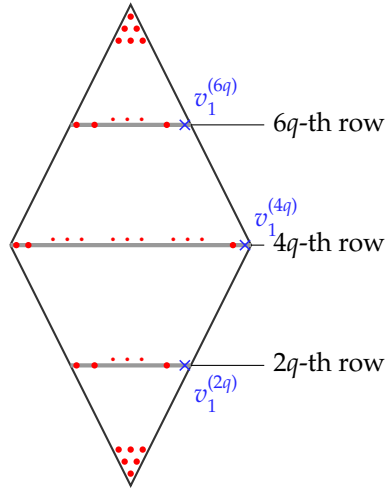

\begin{rmk}
We note that $\mathcal{V}^{\mathbb{U}}_{\left(k;q\right)} \subset \mathbb{R}^{k^2q^2-\left(k-1\right)}$ is full-dimensional.
\end{rmk}

\begin{rmk}
It is immediate that in the simplest possible case $(k,q)=(1,1)$ we have $\mathcal{V}^{\mathbb{U}}_{\left(1;1\right)}=[0,1]$. The next simplest case is $(k,q)=(2,1)$. After some simple computations it can be seen that $\mathcal{V}^{\mathbb{U}}_{\left(2;1\right)}$ is actually a pyramid, see Figure \ref{fig:pyramid} for an illustration. For higher values of the parameters $k,q$, $\mathcal{V}^{\mathbb{U}}_{\left(k;q\right)}$ is at least four-dimensional.
\end{rmk}

The following is readily seen to be true:

\begin{prop}\label{PropLatticePointsConnU}
Let $k,q \in \mathbb{N}$. Then for all $N\in \mathbb{N}$ we have:
$$\textnormal{MoM}_{\mathbb{U}(N)}\left(k;q\right) =  \#\left(\mathbb{Z}^{k^2q^2-\left(k-1\right)}\cap N\mathcal{V}_{\left(k;q\right)}^{\mathbb{U}}\right).$$
\end{prop}

\begin{proof}
By the very construction of $\mathcal{V}_{(k;q)}^{\mathbb{U}}$ and Proposition \ref{unitary_patterns}.
\end{proof}

The following is also immediate:
\begin{lem}
\label{unitary_convex_polytope}
Let $k,q \in \mathbb{N}$. $\mathcal{V}^{\mathbb{U}}_{\left(k;q\right)}$ is a rational convex polytope.
\end{lem}\begin{proof}
It is a rational polytope since it is a bounded region defined by linear inequalities with rational coefficients. It is convex since it is the intersection of convex sets.
\end{proof} 

As mentioned in the introduction the polytopes $\mathcal{V}^{\mathbb{U}}_{\left(k;q\right)}$ are in general not integral:

\begin{prop}\label{NonIntegral}
Let $k,q \in \mathbb{N}$ with $k\ge 2$. Then, $\mathcal{V}^{\mathbb{U}}_{\left(k;q\right)}$ is not an integral polytope.
\end{prop} \begin{proof} We need to exhibit a single non-integral vertex. We claim that the point $\left(\frac{1}{2},\dots, \frac{1}{2}\right)$ is a vertex. For each pair of diagonally adjacent coordinates we have inequality in $\mathcal{V}_{(k;q)}^{\mathbb{U}}$ by the interlacing so we can define a face by setting them equal. For all $l = 1,\ldots,k-1$ consider the fixed element $v_1^{(2ql)}$, see Figure \ref{fig:unitary_k4}. By the interlacing $v_1^{(2ql)} \geq v_2^{(2ql)}$ in $\mathcal{V}_{(k;q)}^{\mathbb{U}}$ so we can define a face by setting $v_1^{(2ql)} = v_2^{(2ql)}$. Strictly speaking, it is the right hand side of (\ref{unitary_fixed}) that we set equal to $v_2^{(2ql)}$. Now consider the intersection of all of these faces. All coordinates as well as the additional (fixed) elements are set equal. Since $k \geq 2$ there is at least one sum constraint. One then sees that the sum constraints are satisfied if and only if the single coordinate is set to $\frac{1}{2}$. Non-empty intersections of faces are faces so this is a zero-dimensional face, namely a vertex.
\end{proof}

\begin{rmk}
In general, this is not the only non-integral vertex. For example, computer calculations show that $\mathcal{V}_{(2;2)}^{\mathbb{U}}$ has $36$ integral vertices and $504$ non-integral vertices.
\end{rmk}

As mentioned in Remark \ref{RemarkReflexive}, if one looks at $4\mathcal{V}^{\mathbb{U}}_{(2;1)}$, it might be tempting to think that $2kq\mathcal{V}^{\mathbb{U}}_{(k;q)}$ is an integral polytope. As the next proposition shows this is in general not true, at least for $k=2$.

\begin{prop}\label{NonIntegralU+}
Let $q\in \mathbb{N}$ with $q\ge 4$. Then, $4q\mathcal{V}^{\mathbb{U}}_{\left(2;q\right)}$ is not an integral polytope.
\end{prop}

\begin{proof}
Note that, since $k=2$ there is only a single sum constraint, which is imposed on row $2q$. Suppose for the sake of contradiction that $4q\mathcal{V}^{\mathbb{U}}_{\left(2;q\right)}$ is integral. We abuse notation and write $v_i^{(j)}$ for the integral coordinates of this dilated polytope. Fix any $t \in \mathbb{N}$ with $q \leq t \leq 2q$. Consider the faces of $4q\mathcal{V}^{\mathbb{U}}_{\left(2;q\right)}$ defined by setting all but the first $t$ entries of row $2q$ to zero. Namely, we set $v_{2q}^{(2q)} = 0$, $\ldots$, $v_{t+1}^{(2q)} = 0$. Then, we further set $v_{t}^{(2q)}= v_{1}^{(2q)}$ (note that due to the interlacing more coordinates are equal; including the entries $v_j^{(2q)}$ for $j=2,\dots,t-1$) and set all remaining entries $v_i^{(j)}$ not already fixed, directly or through the interlacing, to zero as well. This clearly ensures that we will always have a unique point in the intersection. That is, the described collection of faces defines a vertex. Considering the dilated sum constraint we must have: $$tv_1^{(2q)} = 4q^2,$$ so since by assumption $4q\mathcal{V}^{\mathbb{U}}_{\left(2;q\right)}$ is integral $t$ divides $4q^2$. From the fact that this holds for any such $t$ we deduce that: \begin{align}\label{dilation_contradiction}4q^2 \geq \textnormal{lcm}(q,\ldots,2q) = \textnormal{lcm}(1,\ldots,2q) > 2^{2q},\end{align} where the last inequality is true since $q \geq 4$. However, (\ref{dilation_contradiction}) is the sought contradiction.
\end{proof}

Although the polynomiality statement in Theorem \ref{PolyThm} is not an immediate consequence of Ehrhart theory (because of Proposition \ref{NonIntegral} above) it is still possible to give an alternative combinatorial proof to those given in \cite{AssiotisKeating,BaileyKeating}.

\begin{prop}
Let $k,q \in \mathbb{N}$. Then, the counting function
$N \mapsto \#\mathcal{UP}_{\left(k;q;N\right)}$
is a polynomial.
\end{prop}

\begin{proof} Let $M\in \mathbb{N}$. For any signature $\boldsymbol{\lambda} \in \mathcal{S}_M^+$ and non-negative integer vector $\boldsymbol{\mu}\in \mathbb{Z}_{\ge 0}^M$ consider the following function, for $N \in \mathbb{N}$:
\begin{align*}
\mathsf{f}_{\boldsymbol{\lambda},\boldsymbol{\mu}}(N)=\#\bigg\{ \textnormal{unitary Gelfand-Tsetlin patterns with top row } N\boldsymbol{\lambda} \textnormal{ and sum }\\ \textnormal{ of entries of the j-th row equal to } \sum_{i=1}^jN\mu_i   \bigg\}.
\end{align*}
Then, it is known that, see for example Proposition 2.6 in \cite{LoeraMcAllister}, $\mathsf{f}_{\boldsymbol{\lambda},\boldsymbol{\mu}}$ is a polynomial.

We now make the following observation, see \cite{AssiotisKeating,CbetaE} for a proof\footnote{There this observation was used in the opposite direction, in going from the set (\ref{AlternativeSet}) to $\mathcal{UP}_{\left(k;q;N\right)}$.}. For any $k,q,N\in \mathbb{N}$, the set $\mathcal{UP}_{\left(k;q;N\right)}$ is in bijection with the following set:
\begin{align}\label{AlternativeSet}
\bigg\{\textnormal{unitary Gelfand-Tsetlin patterns of length } 2kq \textnormal{ with top row } (N,\dots,N,0,\dots,0) \nonumber\\ 
\textnormal{with } N \textnormal { and } 0 \textnormal{ each appearing } kq \textnormal{ times and sum of entries in the 2jq-th row }\nonumber\\ \textnormal{ equal to } Njq, \textnormal{ for } j=1,\dots, k\bigg\}.
\end{align}
Hence, with $\boldsymbol{\lambda}=(1,\dots,1,0,\dots,0)$, where $1$'s and $0$'s each appear $kq$ times, we obtain:
\begin{align*}
  \#\mathcal{UP}_{\left(k;q;N\right)} = \sum_{\substack{\boldsymbol{\mu}\in \mathbb{Z}_{\ge 0}^{2kq}: \sum_{i=1}^{2jq}\mu_i=jq,\\  j=1,\dots,k}} \mathsf{f}_{\boldsymbol{\lambda},\boldsymbol{\mu}}(N), \ \ \forall N \in \mathbb{N},
\end{align*}
from which the conclusion follows.
\end{proof}

\subsection{Symplectic case preliminaries}
We proceed analogously to the previous section. For $k,q,N \in \mathbb{N}$ we express  $\textnormal{MoM}_{\mathbb{SP}(2N)}(k;q)$ in terms of two combinatorial problems. Firstly, one in terms of certain patterns and then one in terms of lattice points in the dilates of an appropriate polytope.

\begin{defn}\label{sym_def}
A symplectic Gelfand-Tsetlin pattern of length $2M$ is a sequence of $2M$ signatures $\left(\boldsymbol\lambda^{\left(j\right)}\right)_{j=1}^{2M}$ such that:
$$\boldsymbol\lambda^{\left(1\right)} \prec \boldsymbol\lambda^{\left(2\right)} \prec \boldsymbol\lambda^{\left(3\right)}\prec\cdots \prec \boldsymbol\lambda^{\left(2M\right)},$$
and $\boldsymbol\lambda^{\left(j\right)} \in \mathcal{S}_{\left\lceil \frac{j}{2} \right\rceil}^+$. If instead $\boldsymbol\lambda^{(j)} \in \mathcal{W}_{\left\lceil \frac{j}{2} \right\rceil}^+$ we say that $\left(\boldsymbol\lambda^{\left(j\right)}\right)_{j=1}^{2M}$ is a continuous symplectic Gelfand-Tsetlin pattern. 

\end{defn}

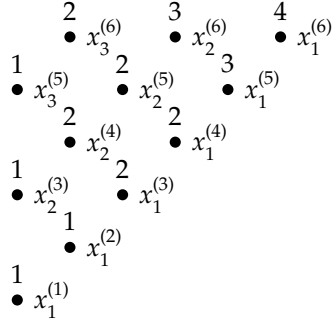
\begin{figure}[H]
\centering
\begin{tikzpicture}[scale = 0.7]
\fill (0,0) circle (3pt) node[right=0.1cm] {$x_1^{(1)}$} node[above=0.1cm] {1}; 
\fill (1,1) circle (3pt) node[right=0.1cm] {$x_1^{(2)}$} node[above=0.1cm] {1}; 
\fill (0,2) circle (3pt) node[right=0.1cm] {$x_2^{(3)}$} node[above=0.1cm] {1}; 
\fill (2,2) circle (3pt) node[right=0.1cm] {$x_1^{(3)}$} node[above=0.1cm] {2}; 
\fill (3,3) circle (3pt) node[right=0.1cm] {$x_1^{(4)}$} node[above=0.1cm] {2}; 
\fill (1,3) circle (3pt) node[right=0.1cm] {$x_2^{(4)}$} node[above=0.1cm] {2}; 
\fill (4,4) circle (3pt) node[right=0.1cm] {$x_1^{(5)}$} node[above=0.1cm] {3}; 
\fill (2,4) circle (3pt) node[right=0.1cm] {$x_2^{(5)}$} node[above=0.1cm] {2}; 
\fill (0,4) circle (3pt) node[right=0.1cm] {$x_3^{(5)}$} node[above=0.1cm] {1}; 
\fill (5,5) circle (3pt) node[right=0.1cm] {$x_1^{(6)}$} node[above=0.1cm] {4}; 
\fill (3,5) circle (3pt) node[right=0.1cm] {$x_2^{(6)}$} node[above=0.1cm] {3}; 
\fill (1,5) circle (3pt) node[right=0.1cm] {$x_3^{(6)}$} node[above=0.1cm] {2}; 
\end{tikzpicture}
\captionsetup{justification=raggedright,singlelinecheck=off}
\caption{An example of a symplectic Gelfand-Tsetlin pattern with $2M=6$ rows.}
\label{fig:symplecticGT_pattern}
\end{figure} 

We have the following analogue of Definition \ref{UnitaryPatterns}:

\begin{defn}
\label{sym_cha}
Let $k,q,N \in \mathbb{N}$. Then, let $\mathcal{SP}_{\left(k;q;N\right)} \subset \mathbb{Z}^{kq(2kq+1)}$, consisting of elements $\left(p_i^{(j)}\right)$, be defined as follows. \begin{enumerate}
    \item For all $(i,j) \in \left\{(m,n) \in \mathbb{N}^2: n \leq 4kq-1, m \leq  kq-\left\lfloor \frac{\left|2kq-n\right|}{2} \right\rfloor \right\}$ we have
    $0 \leq  p_i^{\left(j\right)} \leq N.$
    \item $\left(\boldsymbol p^{\left(j\right)}\right)_{j=1}^{2kq}$ and $\left(\boldsymbol p^{\left(4kq-j\right)}\right)_{j=1}^{2kq}$ form symplectic Gelfand-Tsetlin patterns.
    \item The following constraints, referred to as the sum constraints, are satisfied. If $k$ is even then for $i = 1,\ldots,\frac{k}{2}$ we require:
\begin{align}\label{symplectic_low_constraints}\sum_{j=\left(2i-2\right)q+1}^{\left(2i-1\right)q}\left[\left|\boldsymbol p^{\left(2j\right)}\right|-2\left|\boldsymbol p^{\left(2j-1\right)}\right|+\left|\boldsymbol p^{\left(2j-2\right)}\right|\right] = \sum_{j=\left(2i-1\right)q+1}^{2iq}\left[\left|\boldsymbol p^{\left(2j\right)}\right|-2\left|\boldsymbol p^{\left(2j-1\right)}\right|+\left|\boldsymbol p^{\left(2j-2\right)}\right|\right] \end{align}

\begin{align}\label{symplectic_high_constraints}
\sum_{j=\left(2i-2\right)q+1}^{\left(2i-1\right)q}\left[\left|\boldsymbol p^{\left(4kq-2j\right)}\right|-2\left|\boldsymbol p^{\left(4kq-2j+1\right)}\right|+\left|\boldsymbol p^{\left(4kq-2j+2\right)}\right|\right] &=\nonumber\\
&\hspace{-2.8cm} \sum_{j=\left(2i-1\right)q+1}^{2iq}\left[\left|\boldsymbol p^{\left(4kq-2j\right)}\right|-2\left|\boldsymbol p^{\left(4kq-2j+1\right)}\right|+\left|\boldsymbol p^{\left(4kq-2j+2\right)}\right|\right].
\end{align}

If $k$ is odd we have the same constraints for $i = 1,\ldots,\frac{k-1}{2}$ as well as
\begin{align}\label{symplectic_odd_constraint}\sum_{j=\left(k-1\right)q+1}^{kq}\left[\left|\boldsymbol p^{\left(2j\right)}\right|-2\left|\boldsymbol p^{\left(2j-1\right)}\right|+\left|\boldsymbol p^{\left(2j-2\right)}\right|\right] &=\nonumber \\ \sum_{j=\left(k-1\right)q+1}^{kq}&\left[\left|\boldsymbol p^{\left(4kq-2j\right)}\right|-2\left|\boldsymbol p^{\left(4kq-2j+1\right)}\right|+\left|\boldsymbol p^{\left(4kq-2j+2\right)}\right|\right].\end{align}

\end{enumerate}
In (\ref{symplectic_low_constraints}),(\ref{symplectic_high_constraints}) and (\ref{symplectic_odd_constraint}) above we take the sum of any row that doesn't exist to be zero. Namely, we always let $\left|\boldsymbol p^{\left(0\right)}\right| = \left|\boldsymbol p^{\left(4kq\right)}\right| = 0$. Note that we always have $k$ sum constraints. Finally, we denote the analogous set of arrays with entries strictly between $0$ and $N$ and strict interlacing by $\mathcal{SP}_{\left(k;q;N\right)}^{\neq}$.
\end{defn}

The following was proven in Section 4 of \cite{ABK}.

\begin{prop}\label{sym_pattern}
Let $k,q\in\mathbb{N}$. Then for all $N\in \mathbb{N}$ we have: $$\textnormal{MoM}_{\mathbb{SP}(2N)}\left(k;q\right) = \#\mathcal{SP}_{\left(k;q;N\right)}.$$
\end{prop}
We note that the $k$ sum constraints fix $k$ entries of the pattern in terms of the others. We then choose $k$ entries to fix and define the following index set to keep track of them.\begin{align*}\mathcal{D}_{(k;q)}^{\mathbb{SP}} = \left\{(1,4kq-1)\right\}&\cup\left\{\left(\frac{n}{2},n\right),  n=4q,8q,\ldots,4\left\lfloor\frac{k}{2}\right\rfloor q \right\}\\&\cup\left\{\left(\frac{4kq-n}{2},n\right),n=4\left(\left\lfloor\frac{k}{2}\right\rfloor+1\right)q,4\left(\left\lfloor\frac{k}{2}\right\rfloor+2\right)q,\ldots,4(k-1)q\right\}.\end{align*} See Figure \ref{fig:symplectick4} for a visualisation of $\mathcal{D}_{(k;q)}^{\mathbb{SP}}$. The free variables are then all the others, namely:
\begin{align*}\mathcal{S}_{(k;q)}^{\mathbb{SP}} = \left\{(i,j) \in \mathbb{N}^2:1\le j\leq 4kq-1,1 \le i \leq kq-\left\lfloor\frac{\left|2kq-j\right|}{2}\right\rfloor\right\} \Big\backslash \mathcal{D}^{\mathbb{SP}}_{(k;q)}.\end{align*} 
%\begin{align*}
%\mathcal{S}_{\left(k;q\right)}^{\mathbb{SP}} =& 
%\biggl\{\left(i,j\right) \in \mathbb{N}^2 : i \leq \left\lfloor %frac{j+1}{2} \right\rfloor \textnormal{ and } j \leq 2kq;\\ 
%&\hspace{2cm} \textnormal{or } 1 \leq i \leq \left\lfloor %\frac{4kq-j+1}{2}\right\rfloor \textnormal{ and } 2kq+1 \leq j %\leq 4kq-1; \\
%&\hspace{2cm} j \neq 4q,8q,\ldots,4(k-1)q \biggr\} \\
%&\cup \left\{\left(i,4jq\right) \in \mathbb{N}^2 \colon i \leq %2jq-1 \textnormal{ and } j \leq \left\lfloor %\frac{k}{2}\right\rfloor ; \right.\\
%&\hspace{2cm} \left. \textnormal{or } 1 \leq i \leq %2\left(k-j\right)q-1 \textnormal{ and } \left\lfloor %\frac{k}{2}\right\rfloor +1\leq j < k \right\}.
%\end{align*}
We then have the following analogue of Definition \ref{unitary_polytope}.

\begin{defn}\label{DefSymplPol}
Let $k,q\in \mathbb{N}$. We define $\mathcal{V}_{(k,q)}^{\mathbb{SP}} \subset \mathbb{R}^{kq(2kq+1)-k}$, consisting of elements $\left(v_i^{(j)}\right)$, by the following conditions:

\begin{enumerate}
    \item For $(i,j) \in \mathcal{S}_{(k;q)}^{\mathbb{SP}}$ we have $0 \leq v_i^{(j)} \leq 1$.
    \item We additionally define the following $v_i^{(j)}$, namely for $(i,j) \in \mathcal{D}_{(k;q)}^{\mathbb{SP}}$:\begin{align*}
& v_{\frac{n}{2}}^{\left(n\right)},  \textnormal{ for }n=4q,8q,\ldots,4\left\lfloor\frac{k}{2}\right\rfloor q,\\
 & v_{\frac{4kq-n}{2}}^{\left(n\right)},  \textnormal{ for } n = 4\left(\left\lfloor \frac{k}{2} \right\rfloor+1\right)q,  4\left(\left\lfloor \frac{k}{2} \right\rfloor+2\right)q,\ldots,4\left(k-1\right)q, \\  
 & v_1^{\left(4kq-1\right)},
\end{align*} as the solutions to (\ref{symplectic_low_constraints}), (\ref{symplectic_high_constraints}), (\ref{symplectic_odd_constraint}) and require that each of these also belongs in $[0,1]$.
\iffalse
\begin{align*}
& 0 \leq v_{\frac{n}{2}}^{\left(n\right)} \leq 1,  \textnormal{ for }n=4q,8q,\ldots,4\Bigg\lfloor\frac{k}{2}\Bigg \rfloor q,\\
 & 0 \leq v_{\frac{4kq-n}{2}}^{\left(n\right)} \leq 1,  \textnormal{ for } n = 4\left(\left\lfloor \frac{k}{2} \right\rfloor+1\right)q,\ldots,4\left(k-1\right)q, \\  
 & 0 \leq v_1^{\left(4kq-1\right)} \leq 1.
 \end{align*}
 \fi
 \item Both $\left(\boldsymbol{v}^{\left(j\right)}\right)_{j=1}^{2kq}$ and $\left(\boldsymbol v^{\left(4kq-j\right)}\right)_{j=1}^{2kq}$ form continuous symplectic Gelfand-Tsetlin patterns.
\end{enumerate}
\end{defn}

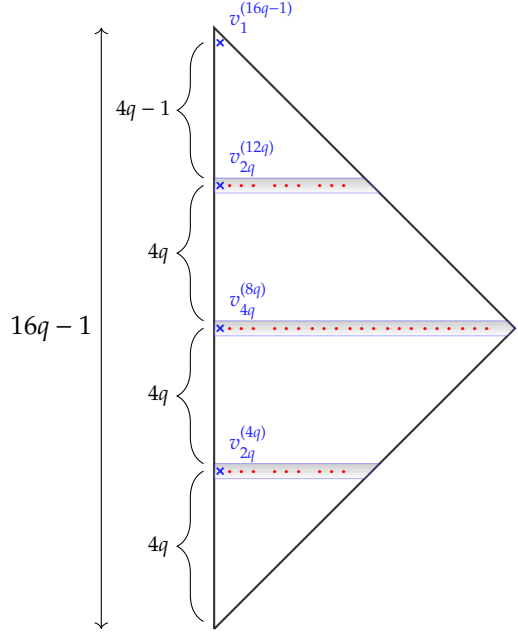
\begin{figure}[H]
\centering
\begin{tikzpicture}
\draw[color=black!80, thick] (0,-4)--(4,0)--(0,4)--(0,-4);

\shade[color=blue,draw,opacity = 0.3] (0,2)--(2,2)--(2.2,1.8)--(0,1.8)--cycle;
\shade[color=blue,draw,opacity = 0.3] (0,0.1)--(3.9,0.1)--(4,0)--(3.9,-0.1)--(0,-0.1)--cycle;
\shade[color=blue,draw,opacity = 0.3] (0,-2)--(2,-2)--(2.2,-1.8)--(0,-1.8)--cycle;
%node
\draw (0.08,3.8) node[line width=0.25mm, cross,blue!85] {} node[above right,blue!85] {\footnotesize $v_1^{(16q-1)}$};;
\draw (0.08,1.9) node[line width=0.25mm, cross,blue!85] {} node[above right,blue!85] {\footnotesize $v_{2q}^{(12q)}$};
\draw (0.08,0) node[line width=0.25mm, cross,blue!85] {} node[above right,blue!85]{\footnotesize$v_{4q}^{(8q)}$};
\draw (0.08,-1.9) node[line width=0.25mm, cross,blue!85] {} node[above right,blue!85] {\footnotesize $v_{2q}^{(4q)}$};

\node[red] (a) at (0.4,1.9) {\ldots};
\node[red] (b) at (1,1.9) {\ldots};
\node[red] (c) at (1.6,1.9) {\ldots};
\node[red] (d) at (.4,0) {\ldots};
\node[red] (e) at (1,0) {\ldots};
\node[red] (f) at (1.5,0) {\ldots};
\node[red] (g) at (2,0) {\ldots};
\node[red] (h) at (2.5,0) {\ldots};
\node[red] (i) at (3,0) {\ldots};
\node[red] (j) at (3.5,0) {\ldots};
\node[red] (k) at (0.4,-1.9) {\ldots};
\node[red] (l) at (1,-1.9) {\ldots};
\node[red] (m) at (1.6,-1.9) {\ldots};
%line
%\draw[color=blue!80,thin] (0,2)--(2,2);
%\draw[color=blue!80,thin] (0,1.8)--(2.2,1.8);
%\draw[color=blue!80,thin] (0,0.1)--(3.9,0.1);
%\draw[color=blue!80,thin] (0,-0.1)--(3.9,-0.1);
%\draw[color=blue!80,thin] (0,-2)--(2,-2);
%\draw[color=blue!80,thin] (0,-1.8)--(2.2,-1.8);
%notation
\draw [decorate,decoration={brace,amplitude=10pt},xshift=-4pt,yshift=0pt] (0,2)--(0,3.8) node [black,midway,xshift=-0.8cm] 
{\footnotesize $4q-1$};
\draw [decorate,decoration={brace,amplitude=10pt},xshift=-4pt,yshift=0pt] (0,0.1)--(0,1.9) node [black,midway,xshift=-0.6cm] 
{\footnotesize $4q$};
\draw [decorate,decoration={brace,amplitude=10pt},xshift=-4pt,yshift=0pt] (0,-1.8)--(0,0) node [black,midway,xshift=-0.6cm] 
{\footnotesize $4q$};
\draw [decorate,decoration={brace,amplitude=10pt},xshift=-4pt,yshift=0pt] (0,-3.9)--(0,-1.9) node [black,midway,xshift=-0.6cm] 
{\footnotesize $4q$};
\draw[black,<->] (-1.5,-4)--(-1.5,4) node[midway,left] {$16q-1$};

\end{tikzpicture}
\captionsetup{justification=raggedright,singlelinecheck=off}
\caption{A visualisation of $\mathcal{D}_{(4;q)}^{\mathbb{SP}}$, corresponding to the four fixed elements depicted as blue crosses that are fixed by solving for the sum constraints (\ref{symplectic_low_constraints}), (\ref{symplectic_high_constraints}), (\ref{symplectic_odd_constraint}). Note that the shaded rows appear in two consecutive sum constraints (\ref{symplectic_low_constraints}), (\ref{symplectic_high_constraints}), (\ref{symplectic_odd_constraint}).}
\label{fig:symplectick4}
\end{figure}

\iffalse
\begin{figure}[H]
\centering
\begin{tikzpicture}
\path[draw=blue!95!black,ultra thick,fill=blue!95!black,fill opacity=0.250000] (0,0)--(2,0)--(2,2)--cycle;
%node
\fill[green!70!black] (0,0) circle (3pt) node[below left] {(0,0)};
\fill[green!70!black] (2,0) circle (3pt) node[below right] {(1,0)};
\fill[green!70!black] (2,2) circle (3pt) node[above right] {(1,1)};
\end{tikzpicture}
\captionsetup{justification=raggedright,singlelinecheck=off}
\caption{Visualisation of $\mathcal{V}_{(1;1)}^{\mathbb{SP}}$}
\label{fig:polytope}
\end{figure}

\fi

\begin{rmk}
It follows from the results in Section 4 of \cite{ABK} that $\mathcal{V}_{\left(k;q\right)}^{\mathbb{SP}}$ is full-dimensional in $\mathbb{R}^{kq\left(2kq+1\right)-k}$. Due to the complexity of the sum constraints this is not as readily evident as in the unitary case.
\end{rmk}

\begin{rmk}
It is not hard to see that in the simplest possible case $(k,q)=(1,1)$ we have that $\mathcal{V}_{\left(1;1\right)}^{\mathbb{SP}}$ is simply a triangle with vertices at $(0,0), (1,0), (1,1)$. For higher values of the parameters $k,q$, $\mathcal{V}^{\mathbb{SP}}_{\left(k;q\right)}$ is at least eight-dimensional.
\end{rmk}

As in the unitary case the following is readily seen to be true:

\begin{prop}\label{PropLatticePointsConnSp}
Let $k,q\in \mathbb{N}$. Then for all $N\in \mathbb{N}$ we have:
$$\textnormal{MoM}_{\mathbb{SP}(2N)}\left(k;q\right) =  \#\left(\mathbb{Z}^{kq\left(2kq+1\right)-k}\cap N\mathcal{V}_{\left(k;q\right)}^{\mathbb{SP}}\right).$$
\end{prop}
\begin{proof}
Follows from Proposition \ref{sym_pattern} by the very construction of $\mathcal{V}_{\left(k;q\right)}^{\mathbb{SP}}$.
\end{proof}

Again the following is also immediate:

\begin{lem}
\label{SP_convex_polytope}
Let $k,q \in \mathbb{N}$. $\mathcal{V}^{\mathbb{SP}}_{\left(k;q\right)}$ is a rational convex polytope.
\end{lem}
\begin{proof}
Same as Lemma \ref{unitary_convex_polytope}.
\end{proof}

Although $\mathcal{V}^{\mathbb{SP}}_{\left(1;1\right)}$ is an integral polytope this is not the case in general as we show next.

\begin{prop}\label{NonIntegralSymplectic}
Let $k,q \in \mathbb{N}$ with $k \geq 2$. Then, $\mathcal{V}^{\mathbb{SP}}_{\left(k;q\right)}$ is not an integral polytope.
\end{prop}

\begin{proof}
We need to exhibit a single non-integral vertex. For $2 \leq j \leq 4kq-2$ define faces by setting $v_1^{(j)} = 1$. Additionally, for $3 \leq j \leq 4kq-3, 2 \leq i \leq  kq-\left\lfloor\frac{\left|2kq-j \right|}{2}\right\rfloor$ define faces by setting $v_i^{(j)} = 0$. Strictly speaking, what we mean by setting a fixed variable to zero is that we solve (\ref{symplectic_low_constraints}), (\ref{symplectic_high_constraints}), (\ref{symplectic_odd_constraint}) for it and then set the resulting equation to zero. Now, consider the intersection of all these faces. We will see that there is a unique point $\boldsymbol v$ in the intersection and that $\boldsymbol v$ is non-integral. Only $v_1^{(1)}$ and $v_1^{(4kq-1)}$ are not immediately fixed through the definitions of the faces as illustrated in Figure \ref{fig:non_integral_vertex}. Recalling the sum constraints (\ref{symplectic_low_constraints}),(\ref{symplectic_high_constraints}), consider the constraints with $i=1$, namely:

\begin{align}\label{first i=1}\sum_{j=1}^{q}\left[\left|\boldsymbol p^{\left(2j\right)}\right|-2\left|\boldsymbol p^{\left(2j-1\right)}\right|+\left|\boldsymbol p^{\left(2j-2\right)}\right|\right] = \sum_{j=q+1}^{2q}\left[\left|\boldsymbol p^{\left(2j\right)}\right|-2\left|\boldsymbol p^{\left(2j-1\right)}\right|+\left|\boldsymbol p^{\left(2j-2\right)}\right|\right],\end{align}

\begin{align}\label{second i=1}
\sum_{j=1}^{q}\left[\left|\boldsymbol p^{\left(4kq-2j\right)}\right|-2\left|\boldsymbol p^{\left(4kq-2j+1\right)}\right|+\left|\boldsymbol p^{\left(4kq-2j+2\right)}\right|\right] &=\nonumber\\
&\hspace{-2.8cm} \sum_{j=q+1}^{2q}\left[\left|\boldsymbol p^{\left(4kq-2j\right)}\right|-2\left|\boldsymbol p^{\left(4kq-2j+1\right)}\right|+\left|\boldsymbol p^{\left(4kq-2j+2\right)}\right|\right].
\end{align} Observe that, the constraints take this form since by assumption $k \neq 1$ \footnote{Note that, in the case $k=1$ there is only a single constraint, see (\ref{symplectic_odd_constraint}), which is not conducive to the remainder of the argument.}. Now, notice that for $1 < j < 4kq-1$ we have $\left|\boldsymbol p^{\left(j\right)}\right| = 1$ so that (\ref{first i=1}), (\ref{second i=1}) become:
\begin{align}
1-2v_1^{(1)}&=0,\label{first_symplectic_vertex_eq}\\
1-2v_1^{(4kq-1)}&=0.\label{second_symplectic_vertex_eq}
\end{align} From (\ref{first_symplectic_vertex_eq}), (\ref{second_symplectic_vertex_eq}) we then get $ v_1^{(1)} = v_1^{(4kq-1)} = \frac{1}{2}$. With this choice for $v_1^{(1)}$ and $v_1^{(4kq-1)}$ we note that every term, namely each quantity in square brackets, of (\ref{symplectic_low_constraints}), (\ref{symplectic_high_constraints}) and (\ref{symplectic_odd_constraint}) vanishes. This implies that the sum constraints are indeed satisfied. We conclude that the intersection of these faces define a vertex and since $v_1^{(1)} = \frac{1}{2}$ is a coordinate of the vertex, it is non-integral.
\end{proof}

\begin{figure}
\centering 
\resizebox{0.3\linewidth}{!}{
\begin{tikzpicture}[scale = 0.7]
\fill (1,0) circle (3pt) node[below right] {0};
\fill (1,2) circle (3pt) node[below right] {0};
\fill (1,4) circle (3pt) node[below right] {1};
\fill (1,-2) circle (3pt) node[below right] {0};
\fill (1,-4) circle (3pt) node[below right] {1};

\fill (2,3) circle (3pt) node[below right] {1};
\fill (2,-3) circle (3pt) node[below right] {1};

\fill (0,3) circle (3pt) node[below right] {0};
\fill (0,-3) circle (3pt) node[below right] {0};
\draw (0,5) node[line width=0.3mm, cross=3pt,blue!85] {} node[above right,blue!85] {$v_1^{(4kq-1)}$};
\draw (0,-5) node[line width=0.3mm, cross=3pt,blue!85] {} node[below right,blue!85] {$v_1^{(1)}$};

\fill (3,0) circle (3pt) node[below right] {0};
\fill (3,2) circle (3pt) node[below right] {1};
\fill (3,-2) circle (3pt) node[below right] {1};

\fill (5,0) circle (3pt) node[below right] {1};

\node (a) at (0,1) {$\vdots$};
\node (b) at (0,-1) {$\vdots$};
\node (c) at (2,1) {$\vdots$};
\node (e) at (2,-1) {$\vdots$};
\node (d) at (4,0) {$\ldots$};
\node[rotate=-45] (f) at (4,1) {$\ldots$};
\node[rotate=45] (g) at (4,-1) {$\ldots$};
\end{tikzpicture}}
\captionsetup{justification=raggedright,singlelinecheck=off}
\caption{The entries in black are immediately fixed by the definitions of the faces in the proof of Proposition \ref{NonIntegralSymplectic}. Assuming $k > 1$, ensures that the entries depicted as blue crosses are also fixed and equal to $\frac{1}{2}$.}
\label{fig:non_integral_vertex}
\end{figure}

\section{Proofs for the unitary group}

We begin by quickly establishing the equality of $\mathsf{P}_{\left(k;q\right)}^{\mathbb{U}}\left(\cdot\right)$ and $\mathcal{L}(\mathcal{V}_{(k;q)}^{\mathbb{U}},\cdot)$.

\begin{prop}\label{mainProU}
Let $k,q \in \mathbb{N}$. Then for all $N \in \mathbb{N}$ we have:
$$\mathsf{P}_{\left(k;q\right)}^{\mathbb{U}}\left(N\right) = \mathcal{L}\left(\mathcal{V}_{\left(k;q\right)}^{\mathbb{U}},N\right).$$
\end{prop}
\begin{proof} First, observe that by combining Proposition \ref{PropLatticePointsConnU}, Lemma \ref{unitary_convex_polytope} and Theorem \ref{Ehrhart_polynomials} we have for all $N \in \mathbb{N}$:
$$\textnormal{MoM}_{\mathbb{U}(N)}\left(k;q\right) = \mathcal{L}\left(\mathcal{V}_{\left(k;q\right)}^{\mathbb{U}},N\right).$$
Then, applying Lemma \ref{vanishing_quasi} gives the statement of the proposition.
\end{proof}

\subsection{Classification of the integer roots}

We will now leverage Theorem \ref{ehrhart_macdonald} to evaluate the moments of moments extension at negative integers, enabling us to find the roots.

\begin{proof}[Proof of the classification of integer roots in Theorem \ref{MainThmU}] Let $k,q \in \mathbb{N}$. By Proposition \ref{mainProU} above and the reciprocity Theorem \ref{ehrhart_macdonald}, we have that for any $N \in \mathbb{N}$:
\begin{align}\label{reciprocityunitary}
  \mathsf{P}_{\left(k;q\right)}^{\mathbb{U}}\left(-N\right) = \left(-1\right)^{k^2q^2-\left(k-1\right)}\overline{\mathcal{L}}\left(\mathcal{V}_{\left(k;q\right)}^{\mathbb{U}},N\right).  
\end{align}
Observe that the dilate of the interior of $\mathcal{V}_{\left(k;q\right)}^{\mathbb{U}}$, $N\textnormal{int}\left(\mathcal{V}_{\left(k;q\right)}^{\mathbb{U}}\right)$, is given by the conditions:
\begin{enumerate}
    \item For $\left(i,j\right) \in \mathcal{S}_{\left(k;q\right)}^{\mathbb{U}},$
$0 < v_i^{\left(j\right)} < N.$
\item For $l = 1,\ldots,k-1$ we define:
 $$v_1^{\left(2ql\right)} = \frac{kq-\left|kq-2ql\right|}{2}- v_2^{\left(2ql\right)}-\cdots- v_{kq-\left|kq-2ql\right|}^{\left(2ql\right)},$$
 and require that
$0 < v_1^{\left(2ql\right)} < N.$
\item $\left(\boldsymbol v^{\left(j\right)}\right)_{j=1}^{kq}$ and $\left(\boldsymbol v^{\left(2kq-j\right)}\right)_{j=1}^{kq}$ both form strict continuous unitary Gelfand-Tsetlin patterns.
\end{enumerate}

Let $N \in \{1,\ldots,2kq-1\}$ and suppose that $\boldsymbol v \in N\textnormal{int}\left(\mathcal{V}_{\left(k;q\right)}^{\mathbb{U}}\right) \cap \mathbb{Z}^{k^2q^2-(k-1)}$. Consider the corresponding strict pattern. The key observation is that, see Figure \ref{fig:unitary_root} for an illustration:

$$v_{kq}^{\left(kq\right)} < v_{kq-1}^{(kq-1)} < \cdots < v_1^{\left(1\right)} < v_{1}^{(2)} < \cdots < v_{1}^{(kq-1)} < v_{1}^{\left(kq\right)},$$
so in particular all $2kq-1$ of these elements are distinct. However, all of these must also lie in the set $\left\{1,\ldots,N-1\right\}$, which has only $N-1$ elements. Hence, for $N \in \left\{1,\ldots,2kq-1\right\}$ by the pigeonhole principle we obtain a contradiction. Thus, no lattice points exist in $N\textnormal{int}\left(\mathcal{V}_{\left(k;q\right)}^{\mathbb{U}}\right)$ for $N \in \left\{1,\ldots,2kq-1\right\}$.  Then, from (\ref{reciprocityunitary}) we obtain that for any such $N$, $\mathsf{P}_{\left(k;q\right)}^{\mathbb{U}}\left(-N\right) = 0.$

Finally, to see that there are no more integer roots we argue as follows. First, observe that for $N \in \mathbb{Z}_{\geq 0}$ the origin lies in $N\mathcal{V}_{(k;q)}^{\mathbb{U}}$ so that $\mathsf{P}_{(k;q)}^{\mathbb{U}}(N) \geq 1$ and then apply the unitary symmetry property (\ref{unitary_symmetry}) which is proven in the following subsection.
\end{proof}

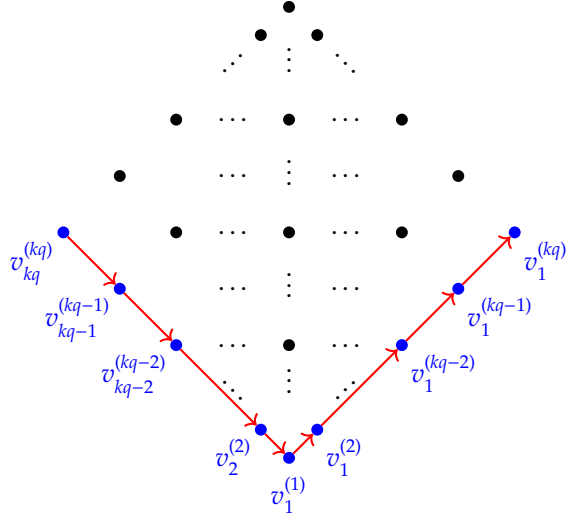
\begin{figure}[H]
\centering
\begin{tikzpicture}[scale = 0.75]
\fill[blue] (0,0) circle (3pt) node[below=0.13cm] {$v_1^{(1)}$};
%\fill (0,0.75) circle (2pt);
\fill[blue] (0.5,0.5) circle (3pt) node[below right] {$v_1^{(2)}$};
\fill[blue] (-0.5,0.5) circle (3pt) node[below left] {$v_2^{(2)}$};
\fill (0,2) circle (3pt);
\fill[blue] (-3,3) circle (3pt) node[below left] {$v_{kq-1}^{(kq-1)}$};
\fill[blue] (-4,4) circle (3pt) node[below left] {$v_{kq}^{(kq)}$};
\fill (0,4) circle (3pt);
\fill[blue] (4,4) circle (3pt) node[below right] {$v_{1}^{(kq)}$};
\fill (0,6) circle (3pt);
\fill (0,8) circle (3pt);
\fill[blue] (3,3) circle (3pt) node[below right] {$v_{1}^{(kq-1)}$};
\fill (3,5) circle (3pt);
\fill (-3,5) circle (3pt);
\fill[blue] (-2,2) circle (3pt) node[below left] {$v_{kq-2}^{(kq-2)}$};
\fill (-2,4) circle (3pt);
\fill (-2,6) circle (3pt);
\fill[blue] (2,2) circle (3pt) node[below right] {$v_{1}^{(kq-2)}$};
\fill (2,4) circle (3pt);
\fill (2,6) circle (3pt);
\fill (-0.5,7.5) circle (3pt);
\fill (0.5,7.5) circle (3pt);

\node (a) at (0,1.5) {$\vdots$};
\node (b) at (0,3.2) {$\vdots$};
\node (c) at (0,5.2) {$\vdots$};
\node (d) at (0,7.2) {$\vdots$};

\node (e) at (-1,2) {$\ldots$};
\node (f) at (-1,4) {$\ldots$};
\node (g) at (-1,6) {$\ldots$};
\node (h) at (1,2) {$\ldots$};
\node (i) at (1,4) {$\ldots$};
\node (j) at (1,6) {$\ldots$};
\node (k) at (1,5) {$\ldots$};
\node (l) at (-1,5) {$\ldots$};
\node (m) at (1,3) {$\ldots$};
\node (n) at (-1,3) {$\ldots$};

\node[rotate=-45] (aa) at (-1,1.2) {$\ldots$};
\node[rotate=-45] (ab) at (1,7) {$\ldots$};
\node[rotate=45] (ac) at (1,1.2) {$\ldots$};
\node[rotate=45] (ad) at (-1,7){$\ldots$};
%line
\draw[red,thick,->] (-3.925,3.925)--(-3.075,3.075); 
\draw[red,thick,->] (-2.925,2.925)--(-2.075,2.075);
\draw[red,thick,->] (-1.925,1.925)--(-0.575,0.575);
\draw[red,thick,->] (-0.425,0.425)--(-0.075,0.075);
\draw[red,thick,->] (0.075,0.075)--(0.425,0.425);
\draw[red,thick,->] (0.575,0.575)--(1.925,1.925);
\draw[red,thick,->] (2.075,2.075)--(2.925,2.925);
\draw[red,thick,->] (3.075,3.075)--(3.925,3.925);
\end{tikzpicture}
\captionsetup{justification=raggedright,singlelinecheck=off}
\caption{With $\boldsymbol{v} \in N\textnormal{int}\left(\mathcal{V}^{\mathbb{U}}_{(k;q)}\right)$ the arrows define a path along which elements are strictly increasing so we observe the labelled elements along the path to be distinct.}
\label{fig:unitary_root}
\end{figure}

\subsection{The unitary symmetry}
Next we prove the symmetry result for the unitary case. The proof proceeds by defining an explicit bijection between the lattice points in a dilation of the interior of $\mathcal{V}_{\left(k;q\right)}^{\mathbb{U}}$ and the lattice points in a less dilated version of the polytope $\mathcal{V}_{\left(k;q\right)}^{\mathbb{U}}$ itself.

\begin{prop} \label{unitary_bijection} Let $k,q,N \in \mathbb{N}$. Then, the map $\mathcal{B}^{\mathbb{U}}_{\left(k;q\right)}$ defined element-wise by:
\begin{align}
 \mathcal{B}^{\mathbb{U}}_{\left(k;q\right)}: \mathcal{UP}_{\left(k;q;N\right)} &\rightarrow \mathcal{UP}_{\left(k;q,N+2kq\right)}^{\neq},\nonumber \\
 \mathcal{B}_{\left(k;q\right)}^{\mathbb{U}}\left(\boldsymbol v\right)_{i}^{\left(j\right)} &= v_i^{\left(j\right)}+2kq+1-\left|kq-j\right|-2i,\label{UbijectionDef}  
\end{align}
is a bijection.
 \end{prop}
 
 \begin{proof}
We first check that $\mathcal{B}_{\left(k;q\right)}^{\mathbb{U}}$ is well-defined. Consider an arbitrary $\boldsymbol v \in \mathcal{UP}_{\left(k;q;N\right)}$. To ensure that $\mathcal{B}_{\left(k;q\right)}^{\mathbb{U}}$ is well-defined, namely that $\mathcal{B}_{\left(k;q\right)}^{\mathbb{U}}\left(\boldsymbol v\right)\in \mathcal{UP}_{\left(k;q;N+2kq\right)}^{\neq}$, we need to check the following conditions. For clarity we also include Figure \ref{fig:unitary_bijection_action} which captures the action of $\mathcal{B}_{\left(k;q\right)}^{\mathbb{U}}$ to each element.
\begin{enumerate}
\item $\mathcal{B}_{\left(k;q\right)}^{\mathbb{U}}\left(\boldsymbol v\right)$ has entries which are in $\{1,\ldots,N+2kq-1\}$. This is immediate since we are starting with entries in $\{0,\ldots,N\}$ and adding elements from $\left\{1,\ldots,2kq-1\right\}$.
\item $\mathcal{B}_{\left(k;q\right)}^{\mathbb{U}}\left(\boldsymbol v\right)$ satisfies strict versions of the interlacing inequalities. This is readily clear from Figure \ref{fig:unitary_bijection_action}, since as we move from one column of $\boldsymbol v$ to the next we add a successively larger positive integer.
\item $\mathcal{B}_{\left(k;q\right)}^{\mathbb{U}}\left(\boldsymbol v\right)$ satisfies the sum-constraints. We simply work through the algebra. The sum constraints satisfied by $\boldsymbol v \in \mathcal{UP}_{(k;q;N)}$ are:
\begin{align*}
 \left| \boldsymbol v^{\left(2ql\right)}\right| = \frac{N}{2}\left(kq-\left|kq-2ql\right|\right), \ \ \textnormal{ for }   l = 1,\ldots,k-1.
\end{align*}
 Then, we can easily compute using (\ref{UbijectionDef}), for $l = 1,\dots,k-1$:
\begin{align*}
   \left|\mathcal{B}_{\left(k;q\right)}^{\mathbb{U}}\left(\boldsymbol v\right)^{\left(2ql\right)}\right| &= \frac{N}{2}\left(kq-\left|kq-2ql\right|\right) + \sum_{i=1}^{kq-\left|kq-2ql\right|}\left(kq+1+kq-\left|kq-2ql\right|-2i\right)  \\
   &= \frac{\left(N+2kq\right)}{2}\left(kq-\left|kq-2ql\right|\right), 
\end{align*}
as required.
\end{enumerate}
Thus, $\mathcal{B}_{\left(k;q\right)}^{\mathbb{U}}$ is well-defined. 

Now, to see that it is injective just observe that it is injective in each element. It remains to show that it is surjective. Take an arbitrary $\boldsymbol u \in \mathcal{UP}_{\left(k;q,N+2kq\right)}^{\neq}$. Then, let $\boldsymbol{v} $ be defined element-wise by the formula:
\begin{align}\label{InverseElementDefU}
v_i^{(j)} = u_i^{\left(j\right)}-2kq-1+\left|kq-j\right|+2i.  
\end{align}
It is clear that $\mathcal{B}_{\left(k;q\right)}^{\mathbb{U}}\left(\boldsymbol v\right) = \boldsymbol u$ so we  only need to show that $\boldsymbol v \in \mathcal{UP}_{\left(k;q;N\right)}$. To do this we need to check the following. 

\begin{enumerate}
    \item $\boldsymbol v$ has entries which are in $\{0,\ldots,N\}$. This can be seen as follows. Due to the strict inequalities in the first condition in the definition of $\mathcal{UP}_{\left(k;q;N+2kq\right)}^{\neq}$ the leftmost entry of $\boldsymbol u$ is bounded from below by $1$. Then, inductively, by the strict interlacing, the entries in the $m$-th column of $\boldsymbol u$ are bounded below by $m$. From the definition (\ref{InverseElementDefU}), see Figure \ref{fig:unitary_bijection_action} for an illustration, this ensures that the entries of $\boldsymbol v$ are bounded from below by $0$. 
    
    Similarly, due to the strict inequalities in the first condition in the definition of $\mathcal{UP}_{\left(k;q;N+2kq\right)}^{\neq}$ the rightmost entry of $\boldsymbol u$ is bounded from above by $N+2kq-1$, Again, inductively, by the strict interlacing, the entries in the $m$-th column of $\boldsymbol u$ are bounded above by $N+m$. From the definition (\ref{InverseElementDefU}), see Figure \ref{fig:unitary_bijection_action} for an illustration, this ensures that the entries of $\boldsymbol v$ are bounded from above by $N$. 
    \item $\boldsymbol v$ satisfies  the (not necessarily strict) interlacing inequalities. This is a consequence of the definition (\ref{InverseElementDefU}), see Figure \ref{fig:unitary_bijection_action} for an illustration, along with the trivial fact that if $a,b \in \mathbb{Z}$ are such that $a > b,$ then for any $c \in \mathbb{N}$,
    $a-c-1 \geq b-c$ (while recalling that $\boldsymbol u$ satisfies strict interlacing).
    
    \item $\boldsymbol v$ satisfies the sum-constraints.  Since $\boldsymbol u$ satisfies the sum constraints we have:
    \begin{align*}
      \left|\boldsymbol u^{(2ql)}\right| = \frac{\left(N+2kq\right)}{2}\left(kq-\left|kq-2ql\right|\right), \ \ \textnormal{ for } l = 1,\dots,k-1.
    \end{align*}
Then, using (\ref{InverseElementDefU}), we can compute, for $l=1,\dots,k-1$:
\begin{align*}
    \left| \boldsymbol v^{\left(2ql\right)}\right| &= \frac{\left(N+2kq\right)}{2}\left(kq-\left|kq-2ql\right|\right) - \sum_{i=1}^{kq-\left|kq-2ql\right|}\left(kq+1+\left(kq-\left|kq-2ql\right|\right)-2i\right)\\ &= \frac{N}{2}\left(kq-\left|kq-2ql\right|\right),
\end{align*}
as required.

\end{enumerate}
\end{proof}

\begin{figure}
\centering
\begin{tikzpicture}[scale = 0.85]

\fill (0,0.5) circle (3pt) node[below right] {\tiny $v_1^{(1)}$} node[above right] {\tiny+kq};
\fill (0,7.5) circle (3pt) node[above right] {\tiny $v_1^{(2kq-1)}$} node[below right] {\tiny+kq};
\fill (0,2) circle (3pt) node[right] {\tiny+kq};
\fill (0,4) circle (3pt) node[right] {\tiny+kq};
\fill (0,6) circle (3pt) node[right] {\tiny+kq};

\fill (2,2) circle (3pt) node[below right] {\tiny$v_1^{(kq-2)}$} node[above left] {\tiny+2kq-3};
\fill (2,6) circle (3pt) node[above right] {\tiny$v_1^{(kq+2)}$} node[below left] {\tiny+2kq-3};
\fill (2,4) circle (3pt) node[above left] {\tiny+2kq-3}; 

\fill (3,3) circle (3pt) node[below right] {\tiny$v_1^{(kq-1)}$} node[below left] {\tiny+2kq-2};
\fill (3,5) circle (3pt) node[above right] {\tiny$v_1^{(kq+1)}$} node[above left] {\tiny+2kq-2};

\fill (4,4) circle (3pt) node[below=0.1cm] {\tiny$v_1^{(kq)}$} node[left] {\tiny+2kq-1};

\fill (-2,2) circle (3pt) node[below left] {\tiny$v_{kq-2}^{(kq-2)}$} node[right] {\tiny+3};
\fill (-2,6) circle (3pt) node[above left] {\tiny$v_{kq+2}^{(kq+2)}$} node[right] {\tiny+3};
\fill (-2,4) circle (3pt) node[right] {\tiny+3}; 

\fill (-3,3) circle (3pt) node[below left] {\tiny$v_{kq-1}^{(kq-1)}$} node[right] {\tiny+2};
\fill (-3,5) circle (3pt) node[above left] {\tiny$v_{kq+1}^{(kq+1)}$} node[right] {\tiny+2};

\fill (-4,4) circle (3pt) node[below=0.1cm] {\tiny$v_{kq}^{(kq)}$} node[right] {\tiny+1};
\node (a) at (0.2,1.4) {$\vdots$};
\node (b) at (0.2,3.1) {$\vdots$};
\node (c) at (0.2,5.1) {$\vdots$};
\node (d) at (0.2,6.8) {$\vdots$};
\node (e) at (-1,2) {$\ldots$};
\node (f) at (-1,4) {$\ldots$};
\node (g) at (-1,6) {$\ldots$};
\node (h) at (1,2) {$\ldots$};
\node (i) at (1,4) {$\ldots$};
\node (j) at (1,6) {$\ldots$};
\node (k) at (1,5) {$\ldots$};
\node (l) at (-1,5) {$\ldots$};
\node (m) at (1,3) {$\ldots$};
\node (n) at (-1,3) {$\ldots$};

\node[rotate=-37] (aa) at (-1,1.25) {$\ldots$};
\node[rotate=-37] (ab) at (1,6.75) {$\ldots$};
\node[rotate=37] (ac) at (1,1.25) {$\ldots$};
\node[rotate=37] (ad) at (-1,6.75) {$\ldots$};

%lines
\draw[blue,style=dashed] (-3,3)--(-3,5);
\draw[blue,style=dashed] (-2,2)--(-2,4)--(-2,6);
\draw[blue,style=dashed] (0,.5) -- (0,2) -- (0,4) -- (0,6) -- (0,7.5);
\draw[blue,style=dashed] (2,2)--(2,4)--(2,6);
\draw[blue,style=dashed] (3,3)--(3,5);

%text 
\node (ta) at (-5.5,4) {\small $1^{st}$ column}; 
\node (tb) at (-5.5,4.7) {\small $2^{nd}$ column}; 
\node (tc) at (-5.5,5.7) {\small $3^{rd}$ column}; 
\node (td) at (-5.5,7.3) {\small $kq^{th}$ column}; 
\node (te) at (5.8,5.7) {\small $(2kq-3)^{th}$ column};
\node (tf) at (5.8,4.7) {\small $(2kq-2)^{th}$ column};
\node (tg) at (5.8,4) {\small $(2kq-1)^{th}$ column};

\draw[red] (ta)--(-4,4);
\draw[red] (tb)--(-3,4.7);
\draw[red] (tc)--(-2,5.7);
\draw[red] (td)--(0,7.3);
\draw[red] (te)--(2,5.7);
\draw[red] (tf)--(3,4.7);
\draw[red] (tg)--(4,4);
\end{tikzpicture}
\captionsetup{justification=raggedright,singlelinecheck=off}
\caption{This illustrates how $ \mathcal{B}^{\mathbb{U}}_{\left(k;q\right)}$, defined element-wise in (\ref{UbijectionDef}), acts on each element. Observe how elements in the same column are all treated the same. When going in the backwards direction, namely in equation (\ref{InverseElementDefU}), each $+$ is replaced by a $-$ everywhere.}
\label{fig:unitary_bijection_action}
\end{figure}
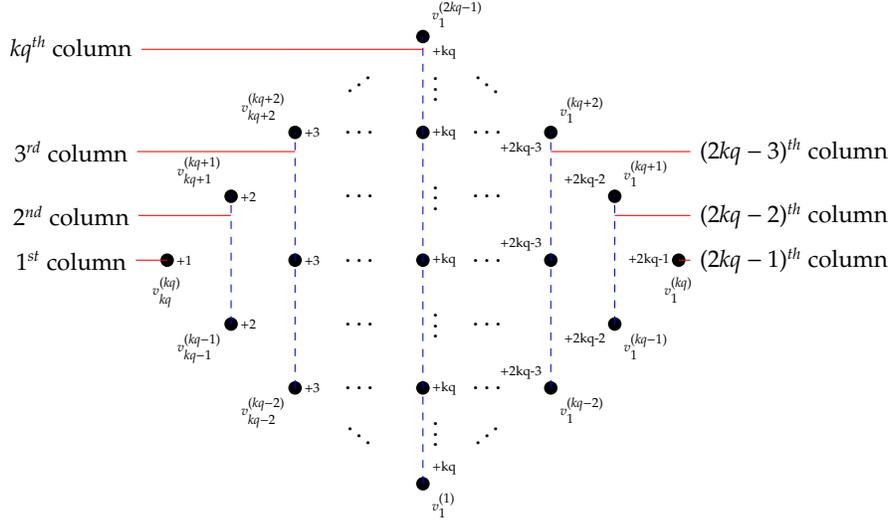

\begin{proof}[Proof of the symmetry property in Theorem \ref{MainThmU}]
We construct a bijection between the lattice points in $N\mathcal{V}_{\left(k;q\right)}^{\mathbb{U}}$ and those in $\left(N+2kq\right)\textnormal{int}\left(\mathcal{V}^{\mathbb{U}}_{\left(k;q\right)}\right)$. Since each lattice point encodes a pattern it is sufficient to construct a bijection between the sets of patterns $\mathcal{UP}_{\left(k;q;N\right)}$ and $\mathcal{UP}_{\left(k;q;N+2kq\right)}^{\neq}$ respectively. Recall that in Proposition \ref{unitary_bijection} we considered the following map defined element-wise by:
\begin{align*}
 \mathcal{B}^{\mathbb{U}}_{\left(k;q\right)}: \mathcal{UP}_{\left(k;q;N\right)} &\rightarrow \mathcal{UP}_{\left(k;q;N+2kq\right)}^{\neq}, \\
 \mathcal{B}_{\left(k;q\right)}^{\mathbb{U}}\left(\boldsymbol v\right)_{i}^{\left(j\right)} &= v_i^{\left(j\right)}+2kq+1-\left|kq-j\right|-2i,
\end{align*}
and we showed that it is a bijection. It being a bijection implies, recalling the reciprocity (\ref{reciprocityunitary}), that for any $N \in \mathbb{Z}_{>2kq}$:

$$\mathsf{P}_{\left(k;q\right)}^{\mathbb{U}}\left(-N\right) = \left(-1\right)^{k^2q^2-\left(k-1\right)}\overline{\mathcal{L}}\left(\mathcal{V}_{\left(k;q\right)}^{\mathbb{U}},N\right)= \left(-1\right)^{k^2q^2-\left(k-1\right)}\mathsf{P}_{\left(k;q\right)}^{\mathbb{U}}\left(N-2kq\right).$$
Writing $N = kq+s$ with $s \in \mathbb{Z}_{>kq}$ this gives:
\begin{align}\label{eqsymU}
\mathsf{P}_{\left(k;q\right)}^{\mathbb{U}}\left(-kq-s\right) = \left(-1\right)^{k^2q^2-\left(k-1\right)}\mathsf{P}_{\left(k;q\right)}^{\mathbb{U}}\left(kq+s-2kq\right)=\left(-1\right)^{k^2q^2-\left(k-1\right)}\mathsf{P}_{\left(k;q\right)}^{\mathbb{U}}\left(-kq+s\right).    
\end{align}
If $k^2q^2-\left(k-1\right)$ is even consider the following polynomial in $s$:
$$g_{\left(k;q\right)}^{\mathbb{U}}\left(s\right) = \mathsf{P}_{\left(k;q\right)}^{\mathbb{U}}\left(-kq-s\right)-\mathsf{P}_{\left(k;q\right)}^{\mathbb{U}}\left(-kq+s\right),$$
while if $k^2q^2-\left(k-1\right)$ is odd consider the polynomial:
$$g_{\left(k;q\right)}^{\mathbb{U}}\left(s\right) = \mathsf{P}_{\left(k;q\right)}^{\mathbb{U}}\left(-kq-s\right)+\mathsf{P}_{\left(k;q\right)}^{\mathbb{U}}\left(-kq+s\right).$$
In either case, notice that by (\ref{eqsymU}), $g_{\left(k;q\right)}^{\mathbb{U}}\left(s\right)$ is a polynomial that vanishes at infinitely many integer points and is thus identically zero. Hence, we have:

$$\mathsf{P}_{\left(k;q\right)}^{\mathbb{U}}\left(-kq-s\right) = \left(-1\right)^{k^2q^2-\left(k-1\right)}\mathsf{P}_{\left(k;q\right)}^{\mathbb{U}}\left(-kq+s\right), \ \ \forall s \in \mathbb{C},$$
and this concludes the proof.
\end{proof}

\section{Proofs for the symplectic group}

We begin by quickly establishing the equality of $\mathsf{P}_{\left(k;q\right)}^{\mathbb{SP}}\left(\cdot\right)$ and $\mathcal{L}\left(\mathcal{V}_{\left(k;q\right)}^{\mathbb{SP}},\cdot\right)$.

\begin{prop}
\label{symplectic_agreement}
Let $k,q\in \mathbb{N}$. Then for all $N \in \mathbb{N}$ we have:

$$\mathsf{P}_{\left(k;q\right)}^{\mathbb{SP}}\left(N\right) = \mathcal{L}\left(\mathcal{V}^{\mathbb{SP}}_{\left(k;q\right)},N\right).$$

\end{prop}
\begin{proof}
First, observe that by combining Proposition \ref{PropLatticePointsConnSp}, Lemma \ref{SP_convex_polytope} and Theorem \ref{Ehrhart_polynomials} we have for all $N \in \mathbb{N}$:
$$\textnormal{MoM}_{\mathbb{SP}(2N)}\left(k;q\right) = \mathcal{L}\left(\mathcal{V}^{\mathbb{SP}}_{\left(k;q\right)},N\right).$$ Then, applying Lemma \ref{vanishing_quasi} gives the statement of the proposition. 
\end{proof}

\subsection{Classification of the integer roots}

We will now leverage Theorem \ref{ehrhart_macdonald} to evaluate the moments of moments extensions at negative integers, enabling us to find the roots.

\begin{proof}[Proof of the classification of integers roots in Theorem \ref{MainThmSp}]
Let $k,q\in \mathbb{N}$. By Proposition \ref{symplectic_agreement} and the reciprocity Theorem \ref{ehrhart_macdonald}, we have that for any $N \in \mathbb{N}$:
\begin{align}\label{reciprocitysymplectic}\mathsf{P}_{\left(k;q\right)}^{\mathbb{SP}}\left(-N\right) = \left(-1\right)^{kq\left(2kq+1\right)-k}\overline{\mathcal{L}}\left(\mathcal{V}_{\left(k;q\right)}^{\mathbb{SP}},N\right). \end{align} Observe that the dilate of the interior of $\mathcal{V}_{(k;q)}^{\mathbb{SP}}$, $N\textnormal{int}\left(\mathcal{V}_{(k;q)}^{\mathbb{SP}}\right)$,  is given by the conditions:

\begin{enumerate}
    \item For $(i,j) \in \mathcal{S}_{(k;q)}^{\mathbb{SP}}$ we have $0 < v_i^{(j)} < N$.
    \item We additionally define the following $v_i^{(j)}$, namely for $(i,j) \in \mathcal{D}_{(k;q)}^{\mathbb{SP}}$: \begin{align*}
& v_{\frac{n}{2}}^{\left(n\right)},  \textnormal{ for }n=4q,8q,\ldots,4\Bigg\lfloor\frac{k}{2}\Bigg \rfloor q,\\
 & v_{\frac{4kq-n}{2}}^{\left(n\right)},  \textnormal{ for } n = 4\left(\left\lfloor \frac{k}{2} \right\rfloor+1\right)q, 4\left(\left\lfloor \frac{k}{2} \right\rfloor+2\right)q, \ldots,4\left(k-1\right)q, \\  
 & v_1^{\left(4kq-1\right)},
\end{align*} as the solutions to (\ref{symplectic_low_constraints}), (\ref{symplectic_high_constraints}), (\ref{symplectic_odd_constraint}) and require that:
\begin{align*}
& 0 < v_{\frac{n}{2}}^{\left(n\right)} < N,  \textnormal{ for }n=4q,8q,\ldots,4\left\lfloor\frac{k}{2}\right \rfloor q,\\
 & 0 < v_{\frac{4kq-n}{2}}^{\left(n\right)} < N,  \textnormal{ for } n = 4\left(\left\lfloor \frac{k}{2} \right\rfloor+1\right)q,4\left(\left\lfloor \frac{k}{2} \right\rfloor+2\right)q,\ldots,4\left(k-1\right)q, \\  
 & 0 < v_1^{\left(4kq-1\right)} < N.
 \end{align*}

 \item Both $\left(\boldsymbol{v}^{\left(j\right)}\right)_{j=1}^{2kq}$ and $\left(\boldsymbol v^{\left(4kq-j\right)}\right)_{j=1}^{2kq}$ form strict continuous symplectic Gelfand-Tsetlin patterns.
\end{enumerate} 

Let $N \in \{1,\ldots,2kq\}$ and suppose $\boldsymbol v \in N\textnormal{int}\left(\mathcal{V}^{\mathbb{SP}}_{\left(k;q\right)}\right)\cap \mathbb{Z}^{kq(2kq+1)-k}$. Consider the corresponding strict pattern. The key observation, see Figure \ref{fig:sym_root} for an illustration, is that:

$$v_1^{\left(1\right)} < v_1^{\left(2\right)} < \cdots < v_1^{\left(2kq\right)},$$
so in particular, these $2kq$ entries are all distinct. However, because of the first point in the characterisation of the interior we also know that all entries are in $\{1,\ldots,N-1\}$ which has $N-1$ elements. When $N \in \left\{1,\dots,2kq\right\}$ observe that this gives a contradiction by the pigeonhole principle and thus there is no such lattice point. Then, from (\ref{reciprocitysymplectic}) we obtain that for such $N$, $\mathsf{P}^{\mathbb{SP}}\left(-N\right) = 0.$

Finally, to see that there are no other integer roots we argue as follows. First observe that for $N \in \mathbb{Z}_{\geq 0}$ the origin lies in $N\mathcal{V}_{(k;q)}^{\mathbb{SP}}$ so that $\mathsf{P}_{(k;q)}^{\mathbb{SP}}(N) \geq 1$ and then apply the symplectic symmetry result (\ref{symplectic_symmetry}) which is proven in the following subsection.
\end{proof}

\begin{figure}[H]
\centering
\begin{tikzpicture}[scale =0.75]
\fill (0,2.4) circle (3pt);
\fill (0,4) circle (3pt);
\fill (0,-2.4) circle (3pt);
\fill[blue] (0,-4) circle (3pt) node[below right] {$v_1^{(1)}$};
%\fill (0.8,0) circle (2pt); 
\fill (0.8,1.6) circle (3pt); 
\fill (0.8,3.2) circle (3pt);
%\fill (0.8,-1.6) circle (3pt);
\fill[blue] (0.8,-3.2) circle (3pt) node[below right] {$v_1^{(2)}$};
\fill (1.6,2.4) circle (3pt);
\fill[blue] (1.6,-2.4) circle (3pt) node[below right] {$v_1^{(3)}$};
%\fill (2.4,1.6) circle (2pt);
\fill (2.4,0) circle (3pt);
\fill (3.6,0.4) circle (3pt);
\fill (0.8,-1.6) circle (3pt);
\fill[blue] (3.6,-0.4) circle (3pt) node[below right] {$v_1^{(2kq-1)}$};
\fill[blue] (4,0) circle (3pt) node[right=0.2cm] {$v_1^{(2kq)}$};
%\node (a) at (0,0.96) {\vdots};
%\node (b) at (0,-0.96) {\vdots};
%\node (c) at (1.6,0.96) {\vdots};
%\node (d) at (1.6,-0.96) {\vdots};
\node (h) at (0.8,1) {$\vdots$};
\node (h) at (0.8,-1) {$\vdots$};
\node (a) at (0,1.5) {$\vdots$};
\node (b) at (0,-1.5) {$\vdots$};
\node (b1) at (0,0) {$\vdots$};
\node (c) at (1.6,0) {$\vdots$};
\node (b2) at (0.8,0) {$\vdots$};
\node[rotate=-45] (e) at (3.2,0.8) {$\ldots$};

\node[rotate=-45] (e) at (2.4,1.6) {$\ldots$};
\node[rotate=45] (l) at (2.4,-1.4) {$\ldots$};
%\node (g) at (3.2,0) {\ldots};
\node[rotate=-45] (aa) at (1.6,0.8) {$\ldots$};
\node[rotate=45] (ab) at (1.6,-0.8) {$\ldots$};
%line
\draw[red,thick,->] (0.063,-3.937)--(0.737,-3.263);
\draw[red,thick,->] (0.863,-3.137)--(1.537,-2.463);
\draw[red,thick,->] (1.663,-2.337)--(3.537,-0.463);
\draw[red,thick,->] (3.663,-0.337)--(3.937,-.063);

\end{tikzpicture}
\captionsetup{justification=raggedright,singlelinecheck=off}
\caption{With $\boldsymbol{v} \in N\textnormal{int}\left(\mathcal{V}^{\mathbb{SP}}_{(k;q)}\right)$ the arrows define a path along which elements are strictly increasing. Thus, we can observe the labelled elements along the path to be distinct.}
\label{fig:sym_root}
\end{figure}
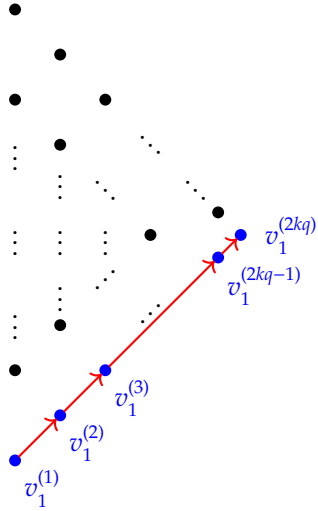

\subsection{The symplectic symmetry}

Next we prove the symmetry result for the symplectic case. As in the unitary case, the proof proceeds by defining an explicit bijection between the lattice points in a dilation of the interior of $\mathcal{V}_{\left(k;q\right)}^{\mathbb{SP}}$ and the lattice points in a less dilated version of the polytope $\mathcal{V}_{\left(k;q\right)}^{\mathbb{SP}}$ itself.

\begin{prop}
\label{symplectic_bijection}
Let $k,q,N \in \mathbb{N}$. Then, the map $ \mathcal{B}^{\mathbb{SP}}_{\left(k;q\right)}$ defined element-wise by:
\begin{align}
 \mathcal{B}^{\mathbb{SP}}_{\left(k;q\right)}: \mathcal{SP}_{\left(k;q;N\right)} &\rightarrow \mathcal{SP}_{\left(k;q,N+2kq+1\right)}^{\neq},\nonumber \\
 \mathcal{B}_{\left(k;q\right)}^{\mathbb{SP}}\left(\boldsymbol v\right)_{i}^{\left(j\right)} &= v_i^{\left(j\right)}+2kq+2-2i-\left|kq-j\right|, \label{SymplecticBijDef}
\end{align}
 is a bijection.
\end{prop}
\begin{proof} We first check that $\mathcal{B}^{\mathbb{SP}}_{(k;q)}$ is well-defined. Consider an arbitrary $\boldsymbol v \in \mathcal{SP}_{(k;q;N)}$. We need to check the following conditions to ensure that $\mathcal{B}^{\mathbb{SP}}_{(k;q;N)}$ is well-defined, namely that $\mathcal{B}^{\mathbb{SP}}_{(k;q)}(\boldsymbol v) \in \mathcal{SP}^{\neq}_{(k;q;N+2kq+1)}$. For clarity we also include Figure \ref{fig:symplectic_bijection_action} which captures the action of $\mathcal{B}^{\mathbb{SP}}_{(k;q)}$ to each element.

\begin{enumerate}
    \item $\mathcal{B}^{\mathbb{SP}}\left(\boldsymbol v\right)$ has entries which are in $\{1,\ldots,N+2kq\}$. This is immediate since we are starting with entries in $\{0,\ldots,N\}$ and adding elements from $\{1,\ldots,2kq\}$.
    \item $\mathcal{B}^\mathbb{SP}\left(\boldsymbol v\right)$ satisfies strict versions of the interlacing inequalities. This is readily clear from Figure \ref{fig:symplectic_bijection_action}, since as we move from one column of $\boldsymbol v$ to the next we add a successively larger positive integer. 
    \item $\mathcal{B}^{\mathbb{SP}}\left(\boldsymbol v\right)$ satisfies the sum-constraints. Consider an arbitrary term in one of the sum-constraints imposed on $\boldsymbol v$, namely one of the quantities in square brackets in either (\ref{symplectic_low_constraints}),(\ref{symplectic_high_constraints}) or (\ref{symplectic_odd_constraint}), for a fixed $j$. Denote it by $[\cdot]_{\boldsymbol v}$ and denote the corresponding term in the constraints on $\mathcal{B}^{\mathbb{SP}}\left(\boldsymbol{v}\right)$ by $[\cdot]_{\mathcal{B}^{\mathbb{SP}}\left( \boldsymbol{v}\right)}$. Then, by studying the action of $\mathcal{B}^{\mathbb{SP}}$ , see Figure \ref{fig:symplectic_bijection_action} for an illustration, we see that:
    $$[\cdot]_{\mathcal{B}^{\mathbb{SP}}\left(v\right)}-[\cdot]_{v} = \sum_{l=1}^j2l-2\sum_{l=1}^j(2l-1)+\sum_{l=1}^{j-1}2l = 0.$$
    Hence, each term is invariant under $\mathcal{B}^{\mathbb{SP}}$ and so $\mathcal{B}^{\mathbb{SP}}\left(\boldsymbol v\right)$ satisfies the sum-constraints.
\end{enumerate}
Thus, $\mathcal{B}^{\mathbb{SP}}$ is well-defined.

To see that it is injective just observe that it is injective in each element. It remains to show that it is surjective. Take an arbitrary $\boldsymbol u \in \mathcal{SP}_{\left(k;q,N+2kq+1\right)}^{\neq}$ and let $\boldsymbol v$ be defined elementwise by:
\begin{align}\label{InverseSymplecticBijDef}
    v_i^{\left(j\right)} = u_i^{\left(j\right)} -2kq-2+2i+\left|kq-j\right|.
\end{align}
Clearly $\mathcal{B}_{(k;q)}^{\mathbb{SP}}\left(\boldsymbol v\right) = \boldsymbol u$ so all that remains to be verified is that $\boldsymbol v \in \mathcal{SP}_{\left(k;q,N\right)}^{\mathbb{SP}}$. Thus, we need to check the following conditions.

\begin{enumerate}
    \item $\boldsymbol v$ has entries which are in $\{0,\ldots,N\}$. The entries in the leftmost column of $\boldsymbol u$ are bounded from below by $1$ and $1$ is subtracted from them in the construction of $\boldsymbol v$ so the leftmost entries of $\boldsymbol v$ are bounded from below by $0$. More generally, an inductive argument using the strict interlacing as in the unitary case gives that all elements are bounded from below by $0$. Finally, an analogous argument starting from the rightmost entry shows that all the entries of $\boldsymbol v$ are bounded from above by $N$.
    \item The appropriate interlacing is satisfied. This is a direct consequence, in completely analogous fashion to the unitary case, of the strict interlacing of $\boldsymbol u$ .
    \item The terms in the sum constraints are invariant under the construction of $\boldsymbol v$ from $\boldsymbol u$ by a completely analogous argument to that for the forward direction.
\end{enumerate}

Thus, $\boldsymbol v \in \mathcal{SP}_{\left(k;q;N\right)}^{\mathbb{SP}}$ and so $\mathcal{B}^{\mathbb{SP}}$ is bijective which concludes the proof.
\end{proof}

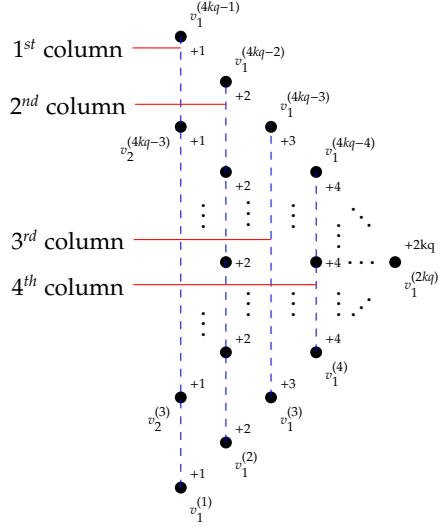
\begin{figure}
\centering
\begin{tikzpicture}[scale =0.75]
\fill (0,2.4) circle (3pt) node[below right] {\tiny +1} node[below left]  {\tiny $v_2^{(4kq-3)}$};
\fill (0,4) circle (3pt) node[below right] {\tiny +1} node[above right] {\tiny$v_1^{(4kq-1)}$};
\fill (0,-2.4) circle (3pt) node[above right] {\tiny+1} node[below left] {\tiny$v_2^{(3)}$};
\fill (0,-4) circle (3pt) node[above right] {\tiny+1} node[below right] {\tiny$v_1^{(1)}$};
\fill (0.8,0) circle (3pt) node[right] {\tiny+2};
\fill (0.8,1.6) circle (3pt) node[below right] {\tiny+2};
\fill (0.8,3.2) circle (3pt) node[below right] {\tiny+2} node[above right] {\tiny$v_1^{(4kq-2)}$};
\fill (0.8,-1.6) circle (3pt) node[above right] {\tiny+2};
\fill (0.8,-3.2) circle (3pt) node[above right] {\tiny+2} node[below right] {\tiny$v_1^{(2)}$};
\fill (1.6,2.4) circle (3pt) node[below right] {\tiny +3} node[above right] {\tiny$v_1^{(4kq-3)}$};
\fill (1.6,-2.4) circle (3pt) node[above right] {\tiny+3} node[below right] {\tiny$v_1^{(3)}$};
\fill (2.4,0) circle (3pt) node[right] {\tiny+4};
\fill (2.4,1.6) circle (3pt) node[below right] {\tiny+4} node[above right] {\tiny$v_1^{(4kq-4)}$};
\fill (2.4,-1.6) circle (3pt) node[above right] {\tiny+4} node[below right] {\tiny$v_1^{(4)}$};
\fill (3.8,0) circle (3pt) node[above right] {\tiny+2kq} node[below right] {\tiny$v_1^{(2kq)}$};
\node (a) at (0.4,0.96) {$\vdots$};
\node (b) at (0.4,-0.96) {$\vdots$};
\node (c) at (2,0.96) {$\vdots$};
\node (d) at (2,-0.6) {$\vdots$};
\node (h) at (1.2,1) {$\vdots$};
\node (h) at (1.2,-0.6) {$\vdots$};
\node (h) at (2.8,0.6) {$\vdots$};
\node (h) at (2.8,-0.6) {$\vdots$};
\node[rotate=-45] (e) at (3.2,0.8) {$\ldots$};
\node[rotate=45] (f) at (3.2,-0.8) {$\ldots$};
\node (g) at (3.2,0) {$\ldots$};
%line
\draw[blue,style=dashed] (0,-4)--(0,-3)--(0,3)--(0,4);
\draw[blue,style=dashed] (0.8,-3.2)--(0.8,-2)--(0.8,2)--(0.8,3.2);
\draw[blue,style=dashed] (1.6,-2.4)--(1.6,2.4);
\draw[blue,style=dashed] (2.4,-1.6)--(2.4,0)--(2.4,1.6);
%text and link
\node (aa) at (-2,3.8) {\small $1^{st}$ column}; 
\node (ab) at (-2,2.8) {\small $2^{nd}$ column}; 
\node (ac) at (-2,0.4) {\small $3^{rd}$ column}; 
\node (ad) at (-2,-0.4) {\small $4^{th}$ column}; 
\draw[red] (aa)--(0,3.8);
\draw[red] (ab)--(0.8,2.8);
\draw[red] (ac)--(1.6,.4);
\draw[red] (ad)--(2.4,-.4);
\end{tikzpicture}
\captionsetup{justification=raggedright,singlelinecheck=off}
\caption{This illustrates how $ \mathcal{B}^{\mathbb{SP}}_{\left(k;q\right)}$, defined element-wise in (\ref{SymplecticBijDef}), acts on each element. Observe how elements in the same column are all treated the same. When going in the backwards direction, namely in equation (\ref{InverseSymplecticBijDef}), each $+$ is replaced by a $-$ everywhere.}
\label{fig:symplectic_bijection_action}
\end{figure}

\begin{proof}[Proof of the symmetry property in Theorem \ref{MainThmSp}] We construct a bijection between the lattice points in  $N\mathcal{V}^{\mathbb{SP}}_{(k;q)}$ and those in $(N+2kq+1)\textnormal{int}\left(\mathcal{V}^{\mathbb{SP}}_{(k;q)}\right)$. Since each lattice point encodes a pattern it is sufficient to construct a bijection between the corresponding sets of patterns $\mathcal{SP}_{\left(k;q;N\right)}$ and $\mathcal{SP}_{\left(k;q;N+2kq+1\right)}^{\neq}$ respectively.
Recall that in Proposition \ref{symplectic_bijection} we considered the following map defined element-wise by:
\begin{align*}
 \mathcal{B}^{\mathbb{SP}}_{\left(k;q\right)}: \mathcal{SP}_{\left(k;q;N\right)} &\rightarrow \mathcal{SP}_{\left(k;q,N+2kq+1\right)}^{\neq}, \\
 \mathcal{B}_{\left(k;q\right)}^{\mathbb{SP}}\left(\boldsymbol v\right)_{i}^{\left(j\right)} &= v_i^{\left(j\right)}+2kq+2-2i-\left|kq-j\right|,
\end{align*}
and we showed that it is a bijection. This gives, also recalling the reciprocity (\ref{reciprocitysymplectic}), that for any $N \in \mathbb{Z}_{>2kq+1}$,
$$\mathsf{P}_{\left(k;q\right)}^{\mathbb{SP}}\left(-N\right) = \left(-1\right)^{kq\left(2kq+1\right)-k}\overline{\mathcal{L}}\left(\mathcal{V}_{\left(k;q\right)},N\right) = \left(-1\right)^{kq\left(2kq+1\right)-k}\mathsf{P}_{\left(k;q\right)}^{\mathbb{SP}}\left(N-2kq-1\right).$$
Writing $N = kq + \frac{1}{2}+s$ where $s\in \mathbb{Z}_{>0}+\frac{1}{2}$ this gives:

$$\mathsf{P}_{\left(k;q\right)}^{\mathbb{SP}}\left(-kq-\frac{1}{2}-s\right) = \left(-1\right)^{kq\left(2kq+1\right)-k}\mathsf{P}_{\left(k;q\right)}^{\mathbb{SP}}\left(-kq-\frac{1}{2}+s\right).$$
Arguing analogously to the unitary case, because $\mathsf{P}_{\left(k;q\right)}^{\mathbb{SP}}$ is a polynomial, this equality extends to all $s\in \mathbb{C}$.

\end{proof}

\bibliographystyle{acm}
\bibliography{referencesMoM}

\bigskip

\noindent{\sc School of Mathematics, University of Edinburgh, James Clerk Maxwell Building, Peter Guthrie Tait Rd, Edinburgh EH9 3FD, U.K.}\newline
\href{mailto:theo.assiotis@ed.ac.uk}{\small theo.assiotis@ed.ac.uk}

\bigskip
\noindent
{\sc School of Mathematics, University of Edinburgh, James Clerk Maxwell Building, Peter Guthrie Tait Rd, Edinburgh EH9 3FD, U.K.}\newline
\href{mailto: e.eriksson-1@sms.ed.ac.uk}{\small  e.eriksson-1@sms.ed.ac.uk}

\bigskip
\noindent
{\sc School of Mathematics, University of Edinburgh, James Clerk Maxwell Building, Peter Guthrie Tait Rd, Edinburgh EH9 3FD, U.K.}\newline
\href{mailto:w.ni-2@sms.ed.ac.uk}{\small w.ni-2@sms.ed.ac.uk}

\end{document}